\documentclass[a4paper,twocolumn,10pt,notitlepage,nofootinbib]{quantumarticle}
\pdfoutput=1
\usepackage[T1]{fontenc}
\usepackage{amsmath}
\usepackage[utf8]{inputenc}
\usepackage[greek,english]{babel}
\usepackage{braket}
\usepackage[dvipsnames]{xcolor}   
\usepackage{xcolor}
\usepackage{graphicx}
\usepackage[format=plain,justification=centerlast,singlelinecheck=true]{caption}
\usepackage[numbers,sort&compress]{natbib}
\usepackage{subcaption}
\usepackage{colortbl}
\usepackage{amssymb}
\usepackage{amsfonts}
\usepackage{hyperref}
\usepackage{enumitem}
\usepackage{amsmath,amsthm}
\usepackage{amsfonts}
\usepackage{amssymb}
\usepackage{graphicx}
\usepackage{hyperref}
\usepackage{bm}
\usepackage{bbm}
\usepackage{color}
\usepackage{mathtools}
\usepackage{multirow}

\newcommand{\poly}[0]{\mathrm{poly}}
\newcommand*{\QEDB}{\hfill\ensuremath{\square}}
\newtheorem{thm}{Theorem}
\newtheorem{corollary}{Corollary}
\newtheorem{lem}{Lemma}

\newtheorem{definition}{Definition}

\newcommand*{\mot}{~}

\newcommand{\proj}[1]{\ket{#1}\bra{#1}}

\begin{document}

\title{Quantum advantage from energy measurements of many-body quantum systems}
\author{Leonardo Novo}
\email{legoncal@ulb.ac.be}
\affiliation{ Centre for Quantum Information and Communication, Ecole Polytechnique de Bruxelles, CP 165, Universit\'e Libre de Bruxelles, 1050 Brussels, Belgium}

\author{Juani Bermejo-Vega}
\email{jbermejovega@go.ugr.es}
\affiliation{ Electromagnetism and Matter Physics Department, University of Granada, Granada, Spain}
\affiliation{ Carlos I Institute of Theoretical and Matter Physics, University of Granada, Granada, Spain}
\affiliation{ Free University of Berlin, Berlin, Germany}

\author{Ra\' {u}l Garc\'{i}a-Patr\'{o}n}
\email{rgarciap@ulb.ac.be}
\affiliation{ Centre for Quantum Information and Communication, Ecole Polytechnique de Bruxelles, CP 165, Universit\'e Libre de Bruxelles, 1050 Brussels, Belgium}

\begin{abstract}
The problem of sampling outputs of quantum circuits has been proposed as a candidate for demonstrating a quantum computational advantage (sometimes referred to as quantum ``supremacy"). In this work, we investigate whether quantum advantage demonstrations can be achieved for more physically-motivated sampling problems, related to measurements of physical observables. We focus on the problem of sampling the outcomes of an energy measurement, performed on a simple-to-prepare product quantum state -- a problem we refer to as energy sampling. For different regimes of measurement resolution and measurement errors, we provide complexity theoretic arguments showing that the existence of efficient classical algorithms for energy sampling is unlikely. In particular, we describe a family of Hamiltonians with nearest-neighbour interactions on a 2D lattice that can be efficiently measured with high resolution using a quantum circuit of commuting gates (IQP circuit), whereas an efficient classical simulation of this process should be impossible. In this high resolution regime, which can only be achieved for Hamiltonians that can be \emph{exponentially fast-forwarded}, it is possible to use current theoretical tools tying quantum advantage statements to a polynomial-hierarchy collapse whereas for lower resolution measurements such arguments fail. Nevertheless, we show that efficient classical algorithms for low-resolution energy sampling can still be ruled out if we assume that quantum computers are strictly more powerful than classical ones. We believe our work brings a new perspective to the problem of demonstrating quantum advantage and leads to interesting new questions in Hamiltonian complexity.
\end{abstract}
%
\maketitle

\section{Introduction}
Impressive recent developments in experimental quantum physics are enabling the manipulation of many-body quantum systems of larger and larger sizes. The high degree of control and local resolution of measurement reached in experimental platforms such as quantum gas microscopes \cite{Bloch2017}, Rydberg atoms manipulated with optical tweezers \cite{Lukin2017}, ion traps \cite{Monroe2017,Lanyon2017}, 
or superconducting circuits \cite{Martinis2018, googlesupremacy2019},
are moving these experiments closer to the quantum advantage frontier -- a regime that is challenging to model using our traditional computers. Experiments at this scale should lead to new insights into important problems in many-body physics.
For example, recent developments of many-body interferometric techniques to estimate the entanglement entropy \cite{Greiner2015,Roos2019} 
have opened experimental access to the investigation of quantum thermalization \cite{Greiner2016}.
Similarly, the access to complex many-body correlators on large-size many-body systems has enabled  the experimental study of quantum critical dynamics and  dynamical phase-transitions \cite{Lukin2017,Lukin2019,Monroe2017}, many-body localization \cite{choi_exploring_2016,Greiner2019},  scrambling \cite{Rey2017,Zoller2019} and topological order \cite{topological1,topological2}. 

Several experimental demonstrations of large-scale quantum simulators that outperform certain classical simulations methods have already been reported\mot\cite{Trotzky,Schreiber-pnas-2015,choi_exploring_2016,Lukin2017,Monroe2017}. Unfortunately, the evidence for quantum advantage in these experiments is based solely on numerical benchmarks against available classical algorithms such as, e.g., DMRG \cite{Trotzky}. Hence, this does not exclude the possibility that a \emph{new classical algorithm} performs as efficiently as a given quantum simulator or quantum algorithm, for a problem where it was previously thought there was an exponential quantum speed-up. A remarkable example where this happened is the recent work ``dequantizing" certain quantum machine learning algorithms \cite{tang2019}.

For this reason, it is of utmost importance to put statements about quantum advantage on rigorous mathematical ground. This has been the subject of several recent works which demonstrate, based on strong complexity-theoretic evidence, that there are certain tasks that can be performed efficiently by quantum devices for which an efficient classical algorithm cannot exist. These are based on sampling problems that exhibit certain robustness against noise and are tailored to near-term hardware. Examples include boson sampling~\cite{bosonsampling}, IQP sampling~\cite{IQPapprox}, sampling from random quantum circuits~\cite{boixo_characterizing_2016, googlesupremacy2019} and quantum simulations of constant-time  Hamiltonian evolutions~\cite{gao,bermejo}. The key strength of these results is that the existence of a quantum advantage is provable assuming plausible complexity theoretic conjectures, such as the non-collapse of the Polynomial Hierarchy (a commonly made assumption in theoretical computer science, which can be seen as a generalization of the P$\neq$NP conjecture)~\cite{bosonsampling,IQPapprox}. The prospect of demonstrating in a reliable way exponential quantum speed-ups has initiated a new field of theoretical and experimental activity coined \emph{quantum computational advantage} (or quantum ``supremacy'') \cite{Qsup,preskill}.

This has motivated several efforts to bring quantum advantage proposals closer to a realistic physical implementation.
These efforts have largely focused on finding approximate sampling problems that are robust against certain experimental errors \cite{bosonsampling,IQPapprox}; tailoring quantum sampling problems to existing implementations \cite{boixo_characterizing_2016,bremner_achieving_2017,bermejo};  verifying such devices with efficient quantum resources \cite{bermejo,gao,miller_quantum_2017} or exponential classical ones \cite{boixo_characterizing_2016,Aaronson:2017:CFQ:3135595.3135617,hangleiter_direct_2017,bouland_quantum_2018}. 
Some works have also brought a many-body physics perspective to previously existing 
quantum advantage proposals, through 
the study of the connections between transitions in sampling complexity and dynamical phase transitions \cite{deshpande_dynamical_2018,muraleedharan_quantum_2019,maskara_complexity_2019}. 
All these works, however, mostly rely on  ``unphysical'' sampling problems that were discovered for the sole purpose of demonstrating a quantum advantage, and further connections with questions of central interest in many-body physics are yet to be explored.

In this work, we take a step further towards identifying physically-motivated quantum advantages
in many-body quantum systems. We ask whether complexity-theoretic results can deliver reliable quantum advantages for the measurement of a physically meaningful observable or the calculation of quantities of physical relevance. Specifically, we investigate whether quantum advantages can be related to the measurement of the energy of a many-body quantum system. Such measurements can be implemented, for example, on a quantum computer via quantum phase estimation \cite{QPEkitaev1995,abrams1999quantum}, or on analog quantum simulators \cite{QNDzoller}.

In retrospective, one could interpret the work by Huh and collaborators \cite{vibronicspectra} as a
first attempt in this direction. This work connects a quantum "supremacy" device, namely a Gaussian boson sampler~\cite{Hamilton2017}, to the problem of determining the vibrational spectrum of a molecule\mot\cite{Barone12}. However, this work did not prove this problem to be hard for a classical computer and, in practice, there exist classical algorithms that build this spectrum for molecules of a hundred harmonic vibrational modes on a desktop within a few minutes~\cite{Santoro2007}.

Hence, demonstrating a conclusive physically motivated quantum advantage remains a difficult milestone. We would like such a quantum advantage proposal to meet two desiderata:
\begin{itemize}
    \item It describes a physical experiment that efficiently measures or estimates
    a relevant many-body observable or quantity.
    \item It provides a rigorous mathematical proof of the impossibility of simulating the outcome of the 
    physical experiment efficiently with classical computers.
\end{itemize}
In this work we focus on the task of sampling outcomes of an energy measurements of a many-body quantum system, a problem we refer to as energy sampling. Naturally, the complexity this task depends on the different parameters characterizing the outcome probability distribution, such as the measurement resolution or other errors affecting the measurement device. Our main contribution is to present complexity theoretic arguments that likely exclude the existence of efficient classical simulators for the energy sampling problem in different parameter regimes. In particular, we demonstrate that for Hamiltonians that can be measured efficiently by quantum devices with very high resolution, it is possible to demonstrate quantum advantage for the energy sampling problem based on the widely believed conjecture of the non-collapse of the Polynomial Hierarchy, together with other standard assumptions \cite{bosonsampling, IQPapprox}. We give an explicit example of a simple family of Hamiltonians (e.g., with nearest neighbour interactions on the 2D square lattice) for which energy measurements are hard to simulate on classical computers, yet, should be relatively feasible to measure on a near-term quantum device, which is able to approximately sample from 2D circuits of commuting gates \cite{bermejo}. Interestingly, for this example, the correct functioning of the quantum measurement device can be efficiently verified using existing fidelity-witness methods~\cite{bermejo}, if reliable single-qubit measurements are available. This leads to a conceivable quantum advantage proposal based on measurements of many-body Hamiltonians. We further discuss limitations of current theoretical tools to prove quantum advantage for energy sampling in low resolution regimes and how it connects to the fundamental problem of proving that quantum computers are strictly more powerful than classical ones. 

For our proof of quantum advantage, we introduce the concept of quantum diagonalizable Hamiltonians. This defines the set of Hamiltonians for which one can efficiently obtain, using a quantum computer, a description of its diagonalizing unitary as a poly-size quantum circuit as well as efficiently compute its eigenvalues. We believe this concept can be of relevance outside the scope of this work and may lead to new investigations of speed-ups with respect to classical algorithms.
 \begin{figure*}
 \centering
 \includegraphics[width=0.8\textwidth]{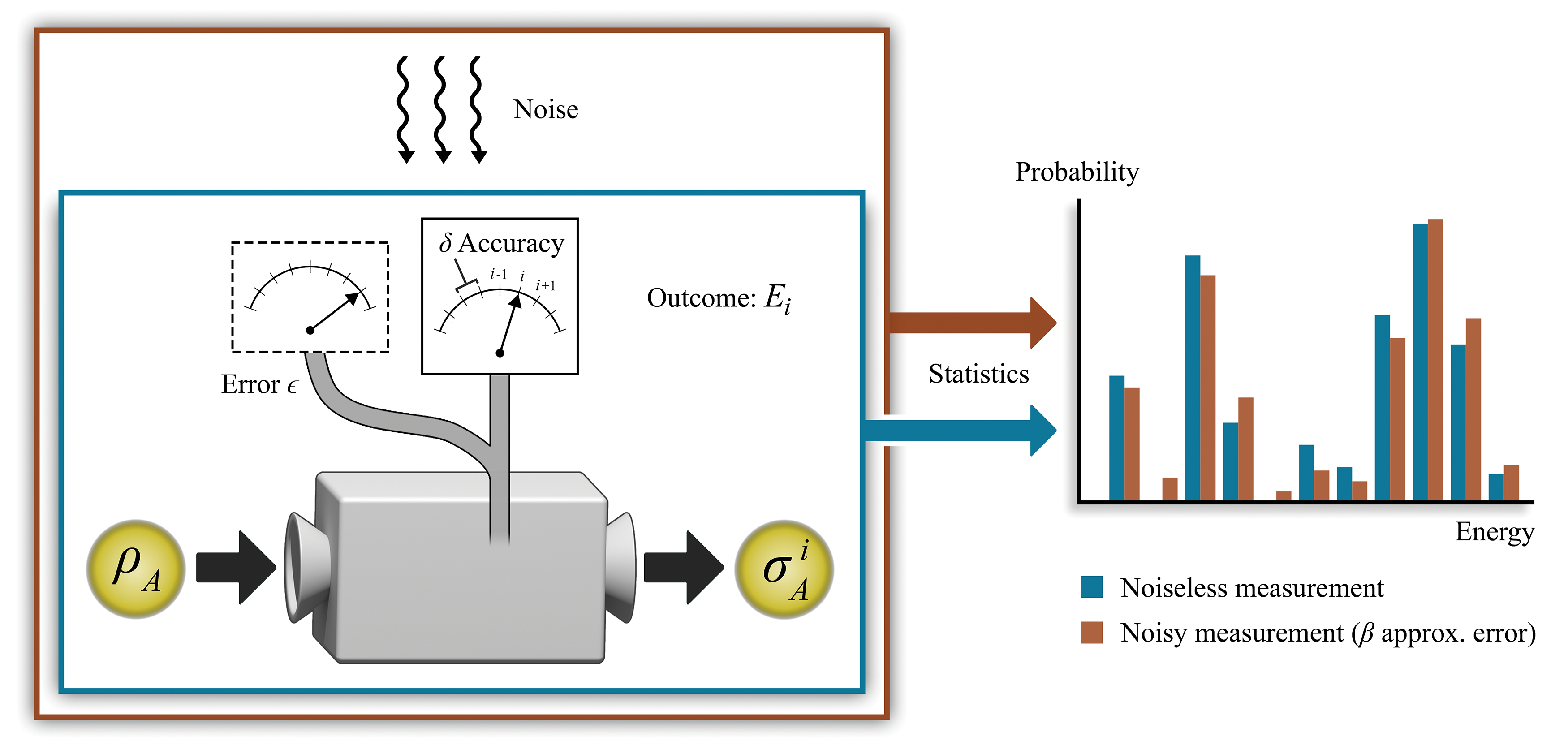}
 \caption{\small{Given a Hamiltonian $H$, an energy measurement can be modelled as a procedure that takes as input a quantum state $\rho_A$ and outputs a measurement outcome $E_i$, as well as a post-measurement state $\sigma_A^i$. We characterize a \emph{noiseless measurement} (inside the smaller blue box) by its measurement resolution $\delta$, which determines the accuracy of the output value $E_i$, and its failure probability $\epsilon$ (see Eqs.~\eqref{eq:measacc} and \eqref{eq:Ensamp}). Moreover, to characterize a \emph{noisy measurement} (inside larger brown box), we introduce an extra parameter $\beta$ that quantifies the sampling error, in total variation distance, between the noiseless probability distribution and the observed one (represented by the blue and brown histogram bars, respectively). These probability distributions are defined explicitly in Definitions~\ref{def:EnergySampling} and \ref{def:approxEnsamp}.} }
 \label{Fig:measurement}
 \end{figure*}
\subsection{Measurement statistics and parameter regimes}\label{sec:terminology}
Before summarizing our contributions in more detail we introduce some terminology regarding the parameters characterizing an energy measurement as well as the different measurement regimes achievable by quantum devices (a more detailed discussion is presented in Sec.~\ref{sec:setting}). 

A model of the measurement outcome statistics needs to take into consideration the imperfect nature of a realistic measurement. A schematic representation of an energy measurement and the parameters we use in this work to characterize it is depicted in Figure \ref{Fig:measurement}. Following Refs.~\cite{wocjan2006BQP,janzing2008BQP,aharonov2017}, we characterize the quality of a measurement by the \textit{measurement resolution} $\mathbf{\delta}$, which
sets the smallest measurement unit, and the \textit{measurement confidence} $\mathbf{\eta}$. For example, an energy measurement of an eigenstate $\ket{\psi_E}$ with energy $E$ is said to have resolution $\delta$ and confidence $\eta$ if it outputs an estimate $E'$ such that
\begin{equation}\label{eq:measacc}
\Pr(|E'-E|\leq \delta )\geq \eta.    
\end{equation}
It will be useful to also define the parameter $\epsilon=1-\eta$, denoting the probability of failure of the measurement. A generalization of Eq.~\eqref{eq:measacc} for arbitrary input states (see Sec.~\ref{sec:setting}) defines the target probability distribution we would like to sample from. 

The finite resolution and measurement confidence result from natural limitations such as a finite measurement time or energy, which are present even if we assume a noiseless measurement device. In addition, to take into account the unavoidable presence of noise in the implementation of a realistic measurement, we introduce the \textit{sampling error} parameter $\mathbf{\beta}$. This parameter quantifies the deviation 
in $\ell_1$-norm between the target outcome distribution and that of an ideal measurement of resolution $\delta$ and confidence $\eta$.

In order to achieve a certain measurement resolution $\delta$, widely used quantum measurement models, such as the von Neumann model or quantum phase estimation (Sec.~\ref{sec:measmodels}), require a 
scaling of the resources needed to perform the measurement which is polynomial in the inverse resolution i.e., $\poly(1/\delta)$. Typically, the resources are quantified by the time required to perform the experiment (assuming a fixed interaction strength between system and a pointer variable \cite{vonneumann}) or by the number of gates of a quantum circuit that implements the desired measurement. We will refer to such a measurement with this performance as a \emph{standard-resolution measurement}, since it can efficiently achieve what we refer to as \emph{standard resolution}, where $\delta=1/\poly(n)$. This can be seen as a coarse-grained energy measurement since, in general, it is not able to distinguish each of the exponentially many eigenvalues. Nevertheless, for unknown Hamiltonians it is the best that can be achieved efficiently. It has been demonstrated that if the Hamiltonian is unknown but its time-evolution can be implemented (as in an experimental setting), an energy-time uncertainty relation is obeyed implying that the measurement time will be inversely proportional to the targeted energy precision \cite{AMPmeasuring2002}. On the other hand, as discussed in \cite{aharonov2017} and in Sec.~\ref{sec:regimes}, in some specific situations one can exploit knowledge of the Hamiltonian to achieve what we refer to as a \emph{super-resolution measurement},
where the scaling of resources is $O(\poly(\log(1/\delta)))$. This allows us to perform a much more accurate measurement which achieves \emph{super-resolution} efficiently, i.e. an  exponentially small measurement resolution $\delta=1/\exp(n)$.

For our purpose it is sufficient to divide the sampling error parameter $\mathbf{\beta}$ into two regimes. We define the measurement as \emph{``approximate''} \cite{bosonsampling,IQPapprox} when the desired sampling error $\beta$ is required to be only some constant independent of the system size $n$. Our results on hardness of approximate sampling extend to the regime $\beta=1/\poly(n)$. Moreover, we define the  \emph{``near-exact''} sampling regime if the sampling error $\beta$ is required to be inverse-exponential in the input size. 
\subsection{Summary of results}\label{sec:summary}
Our results on classical hardness of simulating energy measurements concern the previously defined regimes of resolution and sampling errors as summarized in Table \ref{table:summary_results} and in more detail below. For the sake of clarity we omit the confidence parameter $\eta$, which can be taken to be $\eta=1-O(\beta)$. We provide complexity theoretic evidence that an efficient classical simulation of energy measurements should not be possible, and discuss how the latter provides a suitable test of quantum advantage for suitable resolution and sampling error regimes.  Due to their relevance in describing physical systems, we focus on measurements of \emph{$k$-local Hamiltonians} acting on $n$ qubits (two level quantum systems) i.e., Hamiltonians of the form $H=\sum_j H_j$ where each term $H_j$ acts on $k$ qubits, for constant $k\in O(1)$. Our main contributions are the following:

(i) We provide quantum advantage protocols for  \emph{approximate super-resolution} energy measurements (Sec.~\ref{sec:expaccuracy}). Specifically, we consider Hamiltonians with nearest-neighbor interactions on 2D lattices that can be efficiently diagonalized on a quantum computer. For the latter, we show, first, that approximate super-resolution measurements can be  implemented by building an approximate sampler from the diagonalizing quantum circuit  (Theorem \ref{thm:superres}). At the same time, we prove that these measurements are hard to simulate classically assuming plausible complexity-theoretic conjectures (Corollary \ref{cor:QASuperresolved}). This leads to a verifiable quantum advantage result based on energy measurements that could be feasibly implemented in available quantum simulators. These results exploit a connection between quantum advantage proposals based on simulating constant-time Hamiltonian dynamics \cite{bermejo,gao} and energy measurement problems.
\\
(ii) Super-resolution measurement procedures for \emph{arbitrary} Hamiltonians are unlikely to exist based on complexity theoretic evidence \cite{wocjan2003PSPACE,aharonov2017}. For this reason, we investigate the hardness of energy measurements with \emph{standard resolution}. In Sec.~\ref{sec:standardres_exact}, we give complexity-theoretic evidence that classical computers cannot efficiently simulate energy measurements with standard resolution, even for simple translation-invariant nearest-neighbor Hamiltonians on the square lattice (Theorem~\ref{thm:HardnessInaccurateMeasurement}). Analogously to results obtained in  Refs.\mot\cite{IQP,Morimae14HardnessDQC1,Morimae1704.03640} for other sampling problems, our hardness result is valid in the \emph{near-exact} sampling regime where $\beta=0$ or is inverse-exponential. We give two hardness proofs, one being based on the quantum advantage proposal of \cite{bermejo}, the other being based on  circuit-to-Hamiltonian constructions \cite{kitaev2002}.
\\
(iii) Ideally, one would like a physically motivated
quantum advantage experiment based on \emph{approximate sampling problems with standard resolution}, which are more resilient to imperfections. 
However, in Sec.~\ref{sec:standardres_approx}, we argue that, with current techniques, it is not possible to link the classical hardness of this problem to a Polynomial Hierarchy (PH) collapse as in Refs.\mot\cite{bosonsampling,IQPapprox}. As an intermediate step, we provide alternative hardness results inspired by the BQP-hardness of this problem \cite{wocjan2006BQP,janzing2008BQP}. Using circuit-to-Hamiltonian constructions \cite{kitaev2002}, we show that a hypothetical classical simulator for energy measurement of local Hamiltonians could be used to approximate arbitrary marginals of the output distribution of any poly-sized quantum circuit (Theorem~\ref{thm:polybox}). Based on  hardness of simulating universal quantum circuits \cite{boixo_characterizing_2016,Aaronson:2017:CFQ:3135595.3135617,bouland_quantum_2018}, these results give evidence that approximately measuring a local Feynman-Kitaev Hamiltonian in the standard resolution regime is classically intractable.  As we will discuss in our manuscript, an open challenge in complexity theory would be to tie these hardness results to a Polynomial Hierarchy (PH) complexity-theoretic collapse. Such a result could have additional implications for the development of quantum protocols exhibiting physically-motivated quantum advantages.
\begin{table}
\small
\centering
\setlength{\tabcolsep}{0pt}
\begin{tabular}{c|c|c|}
\cline{2-3}                                                                                                                                         & \begin{tabular}[c]{@{}c@{}}~~Super-resolution~~\\  $ \delta=1/\text{exp}(n)$\end{tabular} & \begin{tabular}[c]{@{}c@{}}~Standard-resolution~\\ $\delta=1/\text{poly}(n)$\end{tabular} \\ \hline
\multicolumn{1}{|c|}{\begin{tabular}[c]{@{}c@{}}~~Near-exact samp.~~\\ $ \beta=1/2^{\text{poly}(n)}$\end{tabular}} & (i) PH-collapse                                                                                       & (ii) PH-collapse                                                                                                \\ \hline

\multicolumn{1}{|c|}{\begin{tabular}[c]{@{}c@{}}~~Approx. samp.~~\\ $ \beta=\text{const.}$~~\end{tabular}}           & \cellcolor[gray]{0.9} (i) PH-collapse*                                                                                       & \cellcolor[gray]{0.9} (iii) BPP=BQP                                                                                                   \\ \hline
\end{tabular}

\caption{\small{Our results on classical hardness for the energy measurement problem, summarized in points (i)-(iii) in Sec.~\ref{sec:summary}. For the different regimes of resolution $\delta$ and sampling error $\beta$, we show the complexity theoretic implications of the existence of an efficient classical algorithm for sampling outcomes of energy measurements, corresponding to the local Hamiltonians we construct. The cells in grey correspond to problems that admit efficient quantum algorithms. In particular, in Sec.~\ref{sec:expaccuracy}, we describe an efficient quantum protocol for approximate super-resolution energy measurements, which could be used to demonstrate a quantum advantage. This result, marked by ``*'', requires plausible complexity theoretic assumptions other than the collapse of the Polynomial Hierarchy (PH).}} 
\label{table:summary_results} 
\end{table}
\section{Setting}\label{sec:setting}
In this section we set up the framework to discuss quantum advantage for measurements of many-body Hamiltonians. We start in Sec.~\ref{sec:measmodels} by discussing two ubiquitous quantum measurement protocols: the Von Neumann pointer, for analog devices, and quantum phase estimation, for digital quantum computers.
We discuss how our ability to measure is limited by experimental noise as well as  physical restrictions on available resources, such as  time,  energy, or quantum gates counts. These limitations motivate us to define the problem of approximate Energy Sampling in Sec.~\ref{sec:ensamp}, where we introduce precisely the parameters mentioned in Sec.~\ref{sec:terminology} characterizing the probability distribution of an imperfect energy measurement. Finally, we discuss in Sec.~\ref{sec:regimes} the different parameter regimes and how they can be achieved by quantum devices.
\subsection{Measurement models and their limitations: from the von Neumann pointer to quantum phase estimation}\label{sec:measmodels}
Let us consider a physical observable $\hat{O}_A$, with eigenvectors $\ket{\psi_i}$
and eigenvalues $\lambda_i$, and a quantum system $A$ in state
$\rho_A$, where both operators act on a finite dimensional Hilbert space $\mathcal{H}$. Upon an ideal measurement of this observable on system $A$, the probability of obtaining an outcome $\lambda_i$ is given by
\begin{equation}\label{eq:perfectmeas}
  p(\lambda_i)={\rm Tr}[\rho_A\Pi_i],
\end{equation}
where $\Pi_i$ is the spectral projection on the eigenstates with eigenvalues $\lambda_i$. In a realistic physical implementation of a quantum measurement, though, the outcome probability distribution deviates from the ideal one and is characterized by a finite resolution and other error parameters.
To understand the fundamental limitations of quantum measurements, let us take as an example the \emph{Von Nenmann pointer} model of quantum  measurement \cite{vonneumann}. In this model, the system under study interacts with a continuous-variable pointer register $R$ for a time $t$ through the unitary coupling
\begin{equation}
    U_{meas}=e^{-i\alpha t\left(\hat{O}_A\otimes\hat{p}_R\right)},
\end{equation}
where $\hat{p}_R$ is the momentum operator of the pointer register and $\alpha$ is a parameter
that captures the strength of the interaction. If $\rho_A$ corresponds to an  
eigenvector $\proj{\psi}_i$ of $\hat{O}_A$, the effect of 
$U_{meas}$ will be to apply a $\alpha\lambda_i t$ continuous shift of the pointer register $R$. Therefore, by having access to the value of the position of the pointer we can infer the outcome  $\lambda_i$. For example, the proposal for quantum non-demolition energy measurements of many-body systems from  \cite{Zoller2019} consists in an implementation of a von Neumann pointer.
Naturally, experimental constraints, such as the finite width of the initial pointer state, a limited accuracy of the measurement of the pointer position, or a finite interaction strength $\alpha$ and interaction time $t$, impose intrinsic limitations on the resolution of the measurement process and the probability that it succeeds. In addition, a realistic measurement process is only able to achieve a noisy approximation of the ideal unitary $U_{meas}$ which results in further errors in the outcome probability distribution.
These limitations justify the introduction of the parameters $\delta$, $\eta$ and $\beta$ presented in Section \ref{sec:summary} (see also Figure \ref{Fig:measurement}) and
rigorously defined in the next subsection.
\begin{figure}
\begin{center}
 \includegraphics[width=\linewidth, trim = 0pt 0cm 0pt 0pt ]{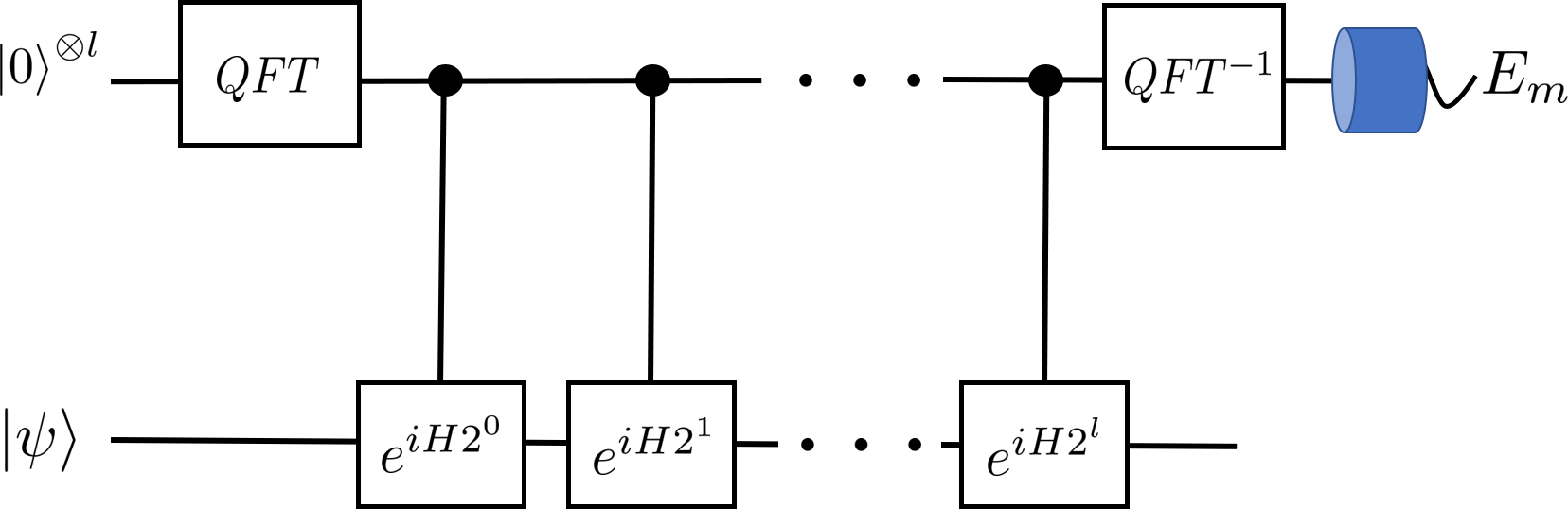}
 \end{center}
  \caption{\small{Scheme of a quantum circuit for energy measurements. An energy measurement on state $\ket{\psi}$ with Hamiltonian $H$ can be performed via the quantum phase estimation (QPE) algorithm, by estimating the eigenphases of the unitary $e^{iH}$. This algorithm exploits a superposition of different time evolutions as well as an efficient implementation of the quantum Fourier transform (QFT).}}
\label{fig:QPEscheme}
\end{figure}
Quantum algorithms for simulating measurements of physical observables also suffer from 
limitations and imperfections. In this case, the measurement resolution as well as the different measurement errors are determined by the finite number of quantum gates, as well as the noise and imperfections in the implementation of these gates.
It is well known that an observable $\hat{O}_{A}$ can be measured using the quantum phase estimation (QPE) algorithm \cite{QPEkitaev1995,abrams1999quantum}, represented in Figure~\ref{fig:QPEscheme}. The QPE algorithm is a quantum algorithm for estimating the eigenphases of a unitary matrix, which can easily be converted into an algorithm for estimating the eigenvalues of $\hat{O}_{A}$ by taking the unitary $\hat{U}=e^{i 2\pi \hat{O}_A/\Lambda}$, where $\Lambda\geq ||\hat{O}_A||$ is an upper bound on the norm of $\hat{O}_{A}$. This way, the eigenphases of $\hat{U}$ can be understood as a normalized outcome of the measurement of $\hat{O}_A$. In fact, it was discussed in Ref.~\cite{childsmeasurement} that the QPE algorithm can be seen as a discretized simulation of a Von-Neumann measurement. In this scenario, the ancillary register of phase-estimation plays the role of a discretized von Neumann pointer, where every ancillary qubit encodes an additional bit of precision. Therefore, a quantum simulation of the measurement process will be also prone to error and necessitate
the introduction of the parameters $\delta$, $\eta$ and $\beta$. The same holds for any potential classical algorithm trying to simulate the quantum measurement process.
\subsection{The Energy Sampling problem}\label{sec:ensamp}
Due to their simplicity and significance in condensed-matter physics, we will focus on measurement problems for local many-body Hamiltonians on two-level systems (qubits). 
Our framework can  be easily extended to measurements of general many-body observables. 
\begin{definition}[\textbf{Measurement resolution}] An energy measurement on a state $\rho$ is said to have resolution $\delta$ and confidence $\eta$ if the probability of outcome $E'$ is such that 
\begin{equation}\label{eq:Ensamp}
\Pr(E'\in[ E_A-\delta , E_B+\delta ])\geq \eta \text{\normalfont{tr}}(\Pi_{[E_A,E_B]}~\rho),
\end{equation}
where $\Pi_{[E_A,E_B]}$ is the spectral projection of $H$ in the energy interval $[E_A,E_B]$.
\end{definition}
We will also define the parameter $\epsilon=1-\eta$, representing the probability of failure of the measurement. 

Without loss of generality,  we assume all observables throughout the text to have spectra contained in $[0,1]$: any observable can be written in this form via a suitable rescaling.\footnote{Let $\kappa$ be an upper bound of $||H||$. Then, $(H':=H/\kappa + 1)/2$ has spectrum in $[0,1]$. Furthermore, such an upper bound can be efficiently computed for $k$-local Hamiltonians $H=\sum_{i=1}^{r} h_i$ with constant $k$, where the number of terms $r$ can be at most $O(n^k)$. Specifically, $\kappa \leq \max_{i}\|h_i\| r$.}  Moreover, although in general an energy $E$ can be any real number, any device performing an energy measurement has a discrete set of output values. For this reason, we discretize the real line into  steps of size $\delta >0$ and assume a measurement outcome is given by a value $E_m\in \{0,\delta ,...,1-\delta, 1\}$  (we take $\delta =1/K$ for some positive integer $K$).\footnote{Strictly speaking the constraints on the measurement outcome distribution from Eq.~\eqref{eq:Ensamp} could allow for outcomes $-\delta$ or $1+\delta$. We assume that these outcomes would be identified with outcome $0$ and $1$, respectively, via classical postprocessing. }

In principle, we could consider energy measurements on any state that can be efficiently prepared by a quantum device. However, we focus on measurements on \emph{product states}, since we are interested in energy measurement problems whose complexity comes from the Hamiltonian and not from the state to be measured. Also, considering measurements on more general quantum states would only increase the classical complexity of the problem. Hence, we define the problem of Energy Sampling as follows.  

\begin{definition}[\textbf{Energy Sampling}]\label{def:EnergySampling} Given an  $n$-qubit product quantum state $\rho$, an $n$-qubit local Hamiltonian $H=\sum_{i=1}^{M} h_i$, $M\in O(\poly(n))$, and parameters $\eta,\delta>0$,   output $E_m\in \{ 0,\delta ...,1-\delta, 1 \}$ with probabilities $q_m$ such that Eq.~\eqref{eq:Ensamp} is fulfilled.
\end{definition}
Such a sampler can be used to build an histogram containing information about how the state $\rho$ decomposes in the eigenbasis of the measured Hamiltonian and thus to learn about the Hamiltonian spectrum, as represented in Figure~\ref{Fig:measurement}. Namely, for a given outcome $E_m$, the probability $p(E_m)$ can be reconstructed up to 1/poly($n$) errors in probabilistic polynomial time. 

It is important to remark that Definition \ref{def:EnergySampling} is not robust to experimental imperfections, as the latter can introduce a sampling error in total variation distance. For this reason, our main interest will be the notion of \emph{approximate} energy sampling, which allows us to consider laboratory errors. 
\begin{definition}[\textbf{$\beta$-approximate Energy Sampling}]\label{def:approxEnsamp}
Given a $n$-qubit product quantum state $\rho$, an $n$-qubit local Hamiltonian $H=\sum_{i=1}^{M} h_i$, $M\in O(\poly(n))$, and parameters $\eta,\delta, \beta >0$, output $E_m\in \{ 0,\delta...,1-\delta,1\}$ with probabilities $q'_m$  such that this probability distribution is $\beta$-close in total variation distance to
the outcome probability distribution of an energy sampler (Definition~\ref{def:EnergySampling}). 
\end{definition}
The parameter $\beta$ quantifies how well the probabilities $q'_m$ approximate the probability distribution of an energy measurement with resolution $\delta$ and confidence $\eta$. Hence, we will refer to the parameter $\beta$ as the sampling error. 
\subsection{Regimes of resolution and error achievable by quantum devices}\label{sec:regimes}
As anticipated in the previous section, theoretical knowledge about the Hamiltonian as well as experimental restrictions lead to different regimes of resolution, confidence and sampling error.
In what follows we extend that presentation with some important remarks.

\emph{Standard-resolution measurements.} For many ubiquitous measurement procedures, the time necessary to achieve resolution $\delta$ grows as ${\rm poly}(1/\delta)$. Taking again as example the von Neumann model, if we assume that the pointer is prepared as a wavepacket of a fixed width $\sigma$ (fixed energy), it is possible to distinguish two consecutive eigenvalues $E_1$ and $E_2$ with high confidence by letting the system interact with the pointer for a time such that $\alpha t |E_2-E_1|\gg\sigma$. Hence, a scaling of $t=O(1/\delta)$ is needed to achieve resolution $\delta$, for a fixed value of the coupling $\alpha$. 

For quantum algorithms that simulate energy measurements, such as QPE, the natural way to quantify their running time is based on the number of quantum gates applied. It is known that an energy measurement protocol based on QPE achieves a resolution $\delta$ with a number of gates scaling as $\poly(1/\delta)$ \cite{abrams1999quantum}. This follows from the fact that the bottleneck for this algorithm is the implementation of an approximation of the time-evolution operator $\exp(i H T)$ for time $T=O(1/\delta)$ in terms of quantum gates, which takes time $\poly(1/\delta, ||H||)$ using standard quantum simulation methods \cite{lloyd1996universal, truncatedtaylor}. More generally, it was shown in Ref.~\cite{aharonov2017} that if the Hamiltonian is unknown and we can only access the Hamiltonian evolution as a black box, or even if its eigenstates are known but there is no information about its eigenvalues, a number of gates scaling as $\poly(1/\delta)$ is the best that can be achieved by any quantum algorithm. 

An alternative way to approximate the probability distribution of an energy measurement on a state $\rho$ is by considering its Fourier transform
\begin{equation}
   P(\omega)~=\sum_i p(\lambda_i)\delta(\omega-\lambda_i)
   =\int\chi(t)e^{i\omega t}dt,
\end{equation}
where $\lambda_i$ are the eigenvalues of the Hamiltonian, $p(\lambda_i)$ is defined in Eq.~\eqref{eq:perfectmeas} and $\chi(t)={\rm Tr}[\rho e^{-i Ht}]$ is the expected value of the time evolution operator. The function $P(\omega)$ can be approximated by measuring the correlators $\chi(t)$ at different times $t$ and calculating the discrete Fourier transform of these values \cite{somma2002,somma2019}. Such correlators can be measured via an many-body-interferometric experiment akin to, for example, Ref.~\cite{Roushanmanybody}.
In this case, if we want to increase the resolution of the approximation by a factor $1/x$, 
it is necessary to measure $x$ times more values of the correlator $\chi(t)$, where the longest time will be $x$ times larger than before. It is then expected that the time needed to run the experiment grows quadratically with the inverse resolution $\delta^{-1}$.

\emph{Super-resolution measurements.} By exploiting certain knowledge about the Hamiltonian it is sometimes possible to construct quantum algorithms for energy measurements whose running time is ${\rm poly}(\log\left(1/\delta\right))$, i.e., the cost of increasing the resolution (decreasing $\delta$) grows polynomially in the number of digits of $\delta$, instead of $\delta$ itself. We say that such a measurement procedure is a \emph{super-resolution measurement}, as it allows to resolve even exponentially small energy gaps of a Hamiltonian, with a cost increasing only polynomially in $n$. It was demonstrated in Ref.~\cite{aharonov2017} that this regime can be achieved by a quantum algorithm iff the corresponding Hamiltonian can be \emph{exponentially fast-forwarded} i.e., the time-evolution $\hat{U}=\exp(-i H T)$ can be implemented for exponentially large $T$ using only polynomially many quantum gates. It is important to note, however, that it is not known how to implement super-resolution measurements of all local or sparse Hamiltonians, and indeed there is strong complexity theoretic evidence that this is impossible \cite{wocjan2003PSPACE,aharonov2017} -- if such a quantum procedure existed, it would imply that quantum computers would be able to solve any problem in PSPACE, which is considered very unlikely. 

\emph{Near-exact sampling.} It is interesting to consider the regime where the sampling error $\beta$ is $0$ or inverse-exponential in $n$ ($\beta=1/2^{\poly(n)}$), in order to understand the hardness of nearly-exact simulations of \emph{ideal} (noiseless) quantum devices~\cite{constantdepth,IQPexact,IQP, bosonsampling}.
As we will show, classical hardness results can be demonstrated in this regime  using the widely believed computational complexity-theoretic conjecture that the Polynomial Hierarchy (PH) does not collapse (see Sec.~\ref{sec:clhardness}).

As an important side remark, we note that achieving the standard resolution and near-exact sampling regime via QPE is non-trivial, even assuming noiseless quantum gates. To achieve this regime, it is necessary to efficiently approximate the time-evolution operator $\hat{U}=e^{i H T}$ up to an inverse-exponential error, which is not possible with methods based on product formulas such as the original quantum simulation proposal of Lloyd\mot\cite{lloyd1996universal}. However, it is possible to solve this problem thanks to recently developed quantum simulation algorithms, which are exponentially more precise than the original proposal by Lloyd, and are applicable for most Hamiltonians of interest such as local, sparse, or low-rank ones \cite{BCK2015,truncatedtaylor,low2017,low2016,chakraborty2018}. 

\emph{Approximate sampling.} The sampling error considered in the near-exact sampling is extremely demanding and it is unknown to be reachable even by fault-tolerant universal quantum devices \cite{aharonov1999fault}.
For this reason, several efforts to develop quantum advantage protocols that are robust against certain sampling errors have been developed \cite{bosonsampling,IQPapprox,Qsup}.
The latter consider approximate sampling problems, where the sampling error $\beta$ is a small constant.
However, such proofs require the introduction of additional computational complexity conjectures, as will be discussed in more detail in Sec.~\ref{sec:conjectures}. Furthermore, we discuss in Secs.~\ref{sec:standardres_approx} that current techniques to demonstrate hardness of approximate sampling fail for energy sampling problems with standard resolution, as they become sampling problems with a small output space.

\emph{Confidence regimes.} For the purposes of our discussion on classical hardness of the energy sampling problem, we can take the value of the measurement failure probability $\epsilon=1-\eta$ to be of the same order of magnitude of the sampling error $\beta$, i.e., in the approximate sampling regime we can tolerate a small constant value of $\epsilon$ whereas in the near-exact sampling regime we require $\epsilon=1/2^{\poly(n)}$.

Using a standard energy measurement procedure such as quantum phase estimation, a resolution $\delta=2^{-l}$ can be achieved with failure probability $\epsilon$ by using $l+\log(2+(2\epsilon)^{-1})$ ancillary qubits \cite{nielsen2002quantum}. The scaling of the number of gates with this approach is $O(\poly(\delta^{-1},\epsilon^{-1}))$ and so a small constant value of $\epsilon$ can be achieved with a constant overhead. Furthermore, any quantum algorithm for energy measurements achieving a failure probability smaller than $\epsilon\leq 1/2-1/\poly(n)$, can be efficiently converted into a procedure achieving an \emph{exponentially small} failure probability $\epsilon=1/2^{\poly(n)}$ via confidence amplification methods~\cite{aharonov2017}. 
\section{Quantum advantage for super-resolution energy measurements}\label{sec:expaccuracy}
As we previously mentioned, \cite{aharonov2017} proves the connection 
between energy measurements with super-resolution and the capability to exponentially fast-forward the time evolution of the Hamiltonian. The intuition behind this result comes from 
the combination of three concepts: (i) the connection between energy measurements and QPE
(ii) the capability to exponentially fast-forward the time evolution of a Hamiltonian; (iii)  the quantum parallelization achieved in QPE that allows exponentially many more queries than a classical
Fourier transform. The authors of \cite{aharonov2017} provide some examples of Hamiltonians amenable to an exponential speed-up of the time evolution:
commuting local Hamiltonians, a Hamiltonian constructed from the modular exponentiation 
unitary in Shor's algorithm and free-fermions. 
The reason behind the capability to exponentially fast-forward the time evolution of free-fermions (which is also applicable in the case of free bosons) 
is the fact that the diagonalization of the Hamiltonian is known \cite{blaizot1986quantum, shchesnovich2013}, 
which allows to construct a quantum circuit that accelerates the simulation of the time evolution. 
Interestingly, this last set of examples can be generalized to define a potentially larger set 
that we name \textit{Quantum Diagonalizable Hamiltonians}.
\subsection{Quantum Hamiltonian Diagonalization}
\label{subsec:QDiag}
In full generality, we say that a Hamiltonian $\tilde{H}$ is quantum diagonalizable if 
\begin{equation}\label{eq:specialHam}
\tilde{H}=U^{\dagger}H_f U, 
\end{equation}
where $U$ is a poly-size quantum circuit and $H_f$ is a diagonal matrix in the computational basis written as
\begin{equation}\label{eq:fofz}
H_f=\sum_{z=0}^{2^n-1} f(z) \ket{z}\bra{z},
\end{equation} 
where the gate decomposition of the circuit $U$ can be obtained efficiently using a quantum computer and there is a quantum circuit that computes $f(z)$ up to a given number of digits of precision $l$ in time $\poly(l)$. The two conditions on $U$ and $f(z)$ guarantee the existence of a poly-size quantum circuit that can exponentially fast-forward the Hamiltonian time evolution (see Figure \ref{fig:FFscheme}).

Hybrid quantum-classical algorithms for finding  quantum circuits for approximate diagonalization of a Hamiltonian have been proposed, as a way to develop more efficient Hamiltonian simulation procedures~\cite{variationalFF}. In this work we will restrict to cases where both the gate decomposition for $U$ and the function $f(z)$ can be computed efficiently \emph{classically}, which allows us to consider simple energy measurement procedures in Sec.~\ref{sec:knowndiag}.
\begin{figure}
\begin{subfigure}{\columnwidth}
  \centering
  \includegraphics[width=0.9\textwidth]{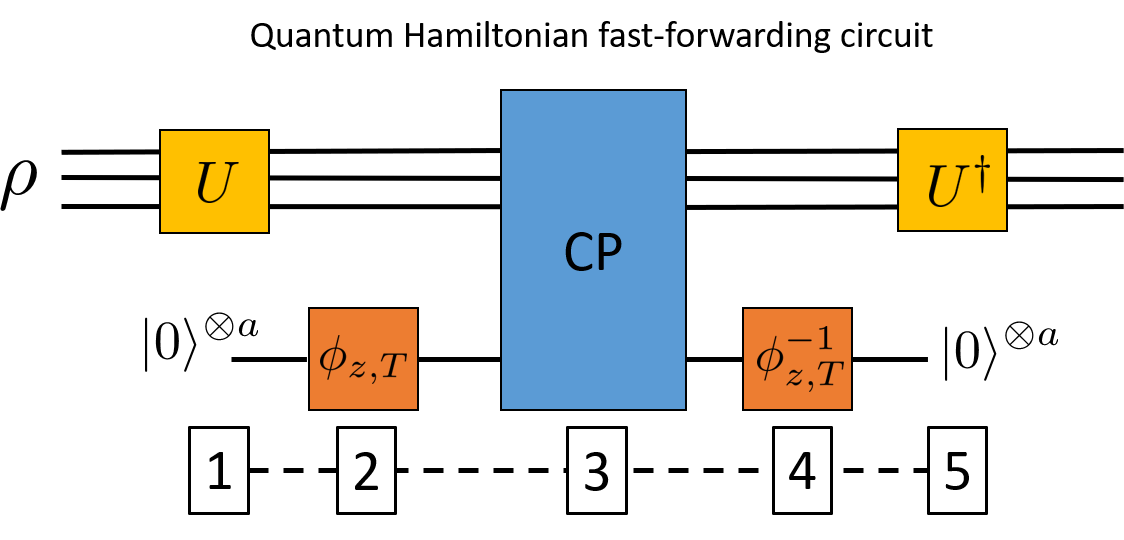}
   \caption{~}
  \label{fig:FFscheme}
\end{subfigure}\\
\begin{subfigure}{\columnwidth}
  \centering
  \includegraphics[width=.75\linewidth]{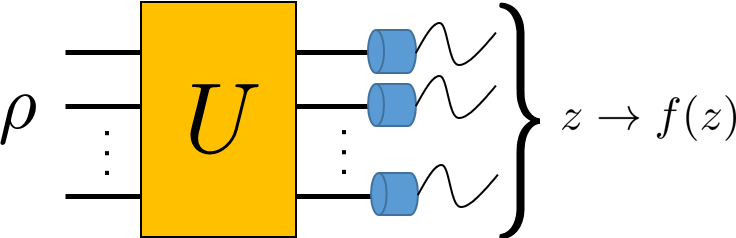}
  \caption{~}
  \label{fig:Usampler}
\end{subfigure}
\caption{\small{a) Representation of the five steps of the quantum circuit for implementing the time-evolution $e^{i \tilde{H} T}$ of a Hamiltonian with a known quantum diagonalization, i.e. of the form $\tilde{H}=U^{\dagger} H_f U$ discussed in Sec.~\ref{subsec:QDiag}. Exploiting the structure of $\tilde{H}$ the time evolution can be implemented efficiently even for exponentially large time $T$. b) If the Hamiltonian has a quantum diagonalization, the energy measurements can be performed by sampling the outcomes $z$ from the quantum circuit implementing $U$ and computing the eigenvalues via the function $f(z)$. This procedure is simpler than QPE and exponentially precise energy measurements can be achieved efficiently.}}
 \label{fig:qdiag}
 \end{figure}
Bellow, we explain how to build a circuit for fast-forwarding a quantum diagonalizable Hamiltonian $\tilde{H}$ by analyzing how the circuit acts on an eigenstate, which has the form $\ket{\psi_z}= U^{\dagger}\ket{z}$. The steps of the circuit are the following:
\begin{enumerate}
    \item Apply $U$ to $\ket{\psi_z}$ to obtain the state $\ket{z}$.
    \item For a given desired evolution time $T$, compute an $a$-bit approximation of the phase $\phi_{z,T}=f(z) T ~(\text{mod}~2\pi)$, where  $a=O(\log(n))$, on an ancillary register. This takes time $O(\poly(\log(T),\log(n)))$ since we assume that $l$-bits of $f(z)$ can be computed in $O(\poly(\log(l)))$ time. This creates a state $\ket{z}\ket{\tilde{\phi}_{z,T}}$, where $\tilde{\phi}_{z,T}$ is the $a$-bit approximation of $\phi_{z,T}$.   
    \item Apply the controlled phase $e^{-i \tilde{\phi}_{z,T}}\ket{z}\ket{\tilde{\phi}_{z,T}}$
    \item Undo the computation of $\tilde{\phi}_{z,T}$ in the ancillary register to obtain the state $e^{-i \tilde{\phi}_{z,T}}\ket{z}\ket{0}^{\otimes a}$.
    \item Apply $U^{\dagger}$ to obtain the state $e^{-i \tilde{\phi}_{z,T}}\ket{\psi_z}$. 
\end{enumerate}
This quantum circuit, sketched in Figure \ref{fig:qdiag}, implements an approximation $U'$ of the time-evolution operator in time $O(\poly(\log(n), \log(T)))$. The choice $a=\log(n)$ ensures that $||U'-\exp(-i \tilde{H} T) ||=O(2^{-a})=O(1/\poly(n))$,  which shows that $\tilde{H}$ is exponentially fast-forwardable, according to the definition in Ref.~\cite{aharonov2017} (see Appendix A).

We leave open the problem of how quantum diagonalizable Hamiltonians relate to the potentially larger class of exponentially fast-forwardable Hamiltonians. Also, it is important to remark that, since the quantum diagonalization assumptions imply the ability to fast-forward, it is unlikely that arbitrary Hamiltonians are quantum diagonalizable -- otherwise, this would imply a general procedure for exponentially precise energy measurements of arbitrary Hamiltonians which, combined with the results in \cite{aharonov2017,wocjan2003PSPACE}, implies BQP=PSPACE. 

\subsection{Super-resolution energy measurements for Hamiltonians with known diagonalization}\label{sec:knowndiag}
The motivation to restrict $f(z)$ to functions that can be computed efficiently classically, instead of the more general case where $f(z)$ can be computed by a quantum circuit, is that it allows us to connect super-resolution energy measurements of $\tilde{H}$ to the problem of sampling from $U$. In fact, given that the quantum diagonalization of the Hamiltonian is known, super-resolution can be achieved by sampling from the quantum circuit $U$, together with some classical post-processing to compute $f(z)$, as schematically depicted in Fig.~\ref{fig:Usampler} . Furthermore, the generated distribution can be characterized by the parameters $\delta, \eta$ and $\beta$ defined in Sec.~\ref{sec:ensamp}. 

In order to show this, let us assume we have access to an approximate sampler from outputs of $U$, i.e. a device that samples approximately from the probabilities $P_z=|\bra{z}U\ket{\psi}|^2$. More precisely, we define the problem of approximate $U$-sampling as follows. 
\begin{definition}[\textbf{$\beta$-approximate $U$-sampling}]
Given an initial state $\ket{\psi}$ and a unitary $U$,
sample outcomes $z$ with probabilities $P'_z$ such that $\sum_z|P_z-P'_z|\leq \beta$. 
\end{definition}
We will refer to such a sampler as a \emph{$\beta$-approximate sampler} for $U$. We can now demonstrate the following 
\begin{thm}[\bf Quantum algorithm for super-resolution energy measurements]\label{thm:superres} 
Consider any quantum diagonalizable Hamiltonian $\tilde{H}=U^{\dagger}H_f U$ as in (\ref{eq:specialHam}). Then, the following quantum algorithm efficiently solves the $\beta-$approximate Energy Sampling problem for Hamiltonian $\tilde{H}$, with the initial state $\ket{\psi}$ and parameters $\eta=1$ and $\delta =2^{-l}$:
\begin{itemize}
\item Query a $\beta$-approximate sampler for $U$, with initial state $\ket{\psi}$.
\item Given an outcome $z$, output an $l$-digit approximation of the value $f(z)$ .
\end{itemize} 
\end{thm}
Theorem~\ref{thm:superres} provides a simple quantum procedure for super-resolution energy measurements, since $l$ digits of precision can be achieved in $\poly(l)$ time. The procedure is represented schematically in Fig.~\ref{fig:qdiag} and can be used to bypass the QPE algorithm, which requires further overhead. The details of the proof of Theorem~\ref{thm:superres} are given in Appendix~\ref{app:proofsthm}, but the result can be understood intuitively. 
 As expected, an $l$-digit accuracy in the computation of $f(z)$ translates into an $l$-digit resolution
 of the energy measurement $\delta$. Moreover, the finite total variation distance $\beta$ of the approximate sampler for $U$ implies that the output distribution of the algorithm solves a $\beta$-approximate energy sampling problem. Finally, since we have assumed $f(z)$ can be computed deterministically, the
confidence $\eta$ is $1$. 

This provides a reinterpretation of the result in \cite{vibronicspectra} as a proposal for super-resolution energy measurements of the vibronic spectra of a molecule (described by a quadratic bosonic Hamiltonian) via a more economical quantum circuit than the traditional quantum phase-estimation algorithm.\footnote{In a nutshell, the Franck-Condon profile at zero temperature which, under certain approximations, gives the vibronic spectrum of a molecule, can be seen as the distribution obtained by measuring the energy of the vaccuum state of a given bosonic quadratic Hamiltonian $H$, according to a different bosonic quadratic Hamiltonian $H'$. If we choose our basis as the Fock basis of $H$, it is possible to compute classically the parameters describing the gaussian transformation $U$ that diagonalizes $H'$ and write $H'= U^{\dagger} H_f U$, i.e. find a quantum diagonalization for $H'$. The proposal of Ref.~\cite{vibronicspectra} for sampling from the Franck-Condon profile can thus be seen as an instance of the scheme for energy measurements discussed in Theorem \ref{thm:superres} and represented in Fig.~\ref{fig:Usampler}.} The discussion above provides also a generalization of that result to any Hamiltonian with a quantum diagonalization.

\subsection{Classical hardness of super-resolution Energy Sampling}\label{sec:clhardness}

\newcommand{\HamiltonianFamily}[1]{H_{2D}\left[#1\right]}
\newcommand{\HamiltonianFamilyDiag}[1]{\tilde{H}_{2D}\left[#1\right]}

\newcommand{\squarelatt}{\mathcal{L}_{\mathrm{2D}}}
In this subsection, we present our quantum advantage result based on super-resolution energy measurements of local Hamiltonians. After introducing the diagonalizable local Hamiltonians of interest for our proof, we show how the existence of an efficient classical (quantum) algorithm for the energy sampling problem implies also an efficient classical (quantum) solution to the problem of sampling from its diagonalization unitary (see Theorem \ref{thm:clhardnesssuperres}). Exploiting known results on the hardness of sampling unitary circuits, we obtain as a corollary (Corollary \ref{cor:QASuperresolved}) the existence of a simple class of local Hamiltonians for which super-resolution measurements can be feasibly implemented in a quantum device but not efficiently classically simulated (assuming plausible complexity theoretic statements).
Our proof is constructive and applies to a family of Hamiltonians diagonalizable by IQP circuits \cite{IQPapprox,bermejo} (i.e.,  quantum circuits that are diagonal in the $X$ Pauli basis). However, the ideas behind the proof are general and can be applied to find other hard energy sampling problems, by considering Hamiltonians that are diagonalized by quantum circuits from other quantum advantage proposals.

\emph{The Hamiltonian.} Specifically, our Hamiltonian family is defined by conjugating a diagonal Hamiltonian of form (\ref{eq:fofz}) with an IQP unitary. For physical reasons, we further focus on a specific family of nearest neighbor IQP circuits on the 2D square lattice $\squarelatt$ \cite{bermejo}. The latter implements a unitary
\begin{equation}\label{U_2D}
U_{2D}=\exp\left[i \frac{\pi}{4}\left(\sum_{(j,k)\in \squarelatt} X_j X_k + \sum_k X_k\right)\right],
\end{equation} 
corresponding to a constant time-evolution of a translation-invariant nearest-neighbor Hamiltonian on $\squarelatt$. The Hamiltonian we use to define an energy measurement problem is of the form \begin{equation}\label{eq:H2d}
\HamiltonianFamily{\vec{w}}=\sum_{j=1}^n w_j U^{\dagger}_{2D} \hat{n}_j U_{2D}, \quad \hat{n}_j := (\mathbb{I}_j-Z_j)/2
\end{equation}
where $\hat{n}_j$ acts locally on qubit $j$ and $\vec{w}=(w_1,...,w_n)$ is a set of real-valued weights. It is easy to see that the unitary (\ref{U_2D}) implements a product  of controlled-Z\mot\cite{nielsen2002quantum} gates in the $X$ basis, whose action on Pauli operators can be analyzed in the stabilizer formalism\mot\cite{GottesmanThesis}. Using this property, we arrive at the final  expression for our Hamiltonian:
\begin{equation}\label{eq:H2dlocal}
\HamiltonianFamily{\vec{w}}= \sum_{j=1}^n w_j X_j \prod_{j:(j,l)\in \squarelatt} Z_l.
\end{equation} 
This Hamiltonian is 5-local, with local terms represented in Fig.~\ref{Fig:lattice}. It is a ``weighted'' non-translation-invariant  variation of the 2D cluster state Hamiltonian used in measurement based quantum computation \cite{Briegel-PRL-2001,RaussendorfPhysRevLett.86.5188}.
\begin{figure}
\includegraphics[scale=0.7]{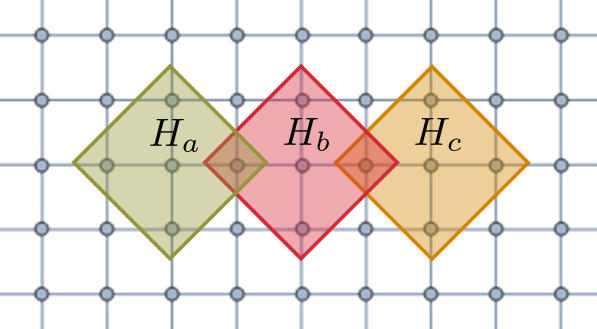}
\caption{ \small{Schematic representation of the Hamiltonian $H_{2D}[\vec{w}]$, from Eq.~\eqref{eq:H2dlocal} and \eqref{eq:H2d}, and three of its local terms $H_a$, $H_b$ and $H_c$.
Each of these terms is a 5-local Hamiltonian acting on a given qubit ($a$, $b$, $c$) and its nearest neighbors on the 2D lattice (represented inside each one of the three squares).
We remark that each term can have a different weight. When all the weights are the same, as in Eq.~\eqref{eq:H2d}, the Hamiltonian is translationally invariant.}}
\label{Fig:lattice}
\end{figure}

\emph{Quantum advantage result.} Our next result gives complexity theoretic evidence that energy measurements of Hamiltonians of form (\ref{eq:H2dlocal}) on product states cannot be efficiently classically simulated. To demonstrate our result, we exploit recent quantum advantage results \cite{bermejo}. Namely, we reduce the problem of measuring these Hamiltonians with super resolution, to that of sampling from the output distribution of a constant-time Hamiltonian evolution of form (\ref{U_2D}) given an input product states of form
\begin{equation}\label{eq:randomstates}
\ket{\psi_{\theta,x}}=\bigotimes_{j=1}^n(\ket{+}+(-1)^{x_j} e^{i \theta_j}\ket{-}),
\end{equation}
where $\theta_j$ and $x_j$ are chosen uniformly at random from the set $\Theta=\{0,\pi/4\}$ and the set $\{0,1\}$, respectively \cite{bermejo}\footnote{We assume measurements are performed in the standard computational basis.}. Precisely, we prove the following reduction between these two problems.

\begin{thm}[Circuit sampling from  energy sampling]\label{thm:clhardnesssuperres}
Consider the problem of measuring 
the Hamiltonians $\HamiltonianFamily{\vec{w}}$ on input product states $\ket{\psi_{\theta,x}}$ with $\theta_j$ and $x_j$ are chosen uniformly at random from the set $\Theta=\{0,\pi/4\}$ and the set $\{0,1\}$, respectively. Assume the existence of an efficient classical  (quantum) algorithm for the associated approximate energy sampling problem with resolution $\delta =O(2^{-n})$, confidence $\eta=1-\epsilon$ and  approximation error $\beta$. Then, there exists an efficient classical (quantum) algorithm for $\gamma-$approximate sampling from the unitary $U_{2D}$ (\ref{U_2D}), acting on the same inputs, with $\gamma=2 \epsilon+\beta$.
\end{thm}
Ref.\mot\cite{bermejo} rules out the existence of efficient classical simulations for short-time Hamiltonian evolutions of form (\ref{U_2D}) based on three plausible complexity-theoretic conjectures: (C1) the non-collapse of the Polynomial Hierarchy; (C2) an approximate average-case hardness conjecture; (C3) an anticoncentration conjecture. These conjectures are  similar to those in \cite{bosonsampling,IQPapprox} and are reviewed in Sec.~\ref{sec:conjectures}. Here, we demonstrate a hardness result for super-resolution Hamiltonian measurements, which is a corollary of Theorem \ref{thm:clhardnesssuperres} and the hardness results in Refs.\mot\cite{bermejo}.
\begin{corollary}[Quantum advantage for approximate super-resolution energy measurements]\label{cor:QASuperresolved}
There cannot exist an efficient classical algorithm for simulating measurements of the Hamiltonian (\ref{eq:H2dlocal}) with super-resolution on input product states as in Theorem \ref{thm:clhardnesssuperres}, for any $\eta=1-\epsilon$, $\beta$ such that $2\epsilon +\beta\leq 1/22$, assuming  complexity-theoretic conjectures C1-C3  below, in section \ref{sec:conjectures}.
\end{corollary}
Corollary \ref{cor:QASuperresolved} leads to a natural quantum advantage proposal, since the problem can be solved via the quantum algorithm in Theorem \ref{thm:superres}, by realizing the unitary (\ref{U_2D}), which can conceivably be implemented in several quantum simulation platforms, e.g., cold-atomic ones \cite{bermejo}. Furthermore, an efficient \emph{quantum verification} protocol for the required quantum sampler exists, which only requires reliable single-qubit measurements\mot\cite{bermejo,hangleiter_direct_2017}. Additionally, the result provides a connection between the quantum advantage protocol of \cite{bermejo}, based on sampling measurements of a quantum state in a fixed basis, and a high-precision spectroscopy problem. This relates these  Hamiltonian quantum advantage proposals to a physical problem of interest, complementing  previous work on vibronic spectra~\cite{vibronicspectra}. 

To prove Theorem\mot\ref{thm:clhardnesssuperres}, we choose an appropriate set of weights $\vec{w}$ that establishes a one-to-one map between the eigenvalues of $H_{2D}[\vec{w}]$ and its eigenvectors. The proof technique (detailed next) generally yields quantum advantage results for measurements of Hamiltonians that are diagonalized by ``Ising-type'' evolutions that implement IQP circuits\mot\cite{IQP,IQPexact,IQPapprox}, if the associated diagonalizing unitary $U$ is hard to simulate classically. For instance, we could replace  $U_{2D}$ with  the long-ranged local IQP circuits of \cite{IQPapprox}, or the nearest-neighbor translation invariant Hamiltonian evolutions of \cite{bermejo,Hangleiter2018anticoncentration,gao}. 
The resulting Hamiltonian for an IQP circuit on a $k$-degree interaction graph would be $k+1$ local (see Appendix~\ref{app:hamcalculations} for details). These alternative constructions yield quantum advantage results analogous to Corollary~\ref{cor:QASuperresolved} for Hamiltonians that are $n$-body long-ranged  considering the quantum circuits from \cite{IQPapprox}; 6-local nearest-neighbor on the dangling-bond square lattice for \cite{bermejo}; and 4 or 5-local for the brickwork lattices of\mot\cite{gao,Hangleiter2018anticoncentration}. Alternatively, the diagonalizing unitaries could also be chosen from other families of quantum circuits that lead to quantum advantage beyond IQP circuits, such as random quantum circuits \cite{boixo_characterizing_2016}. In the latter case, however, the resulting Hamiltonian family would be $n$-local and hence of less physical interest.
\\
\emph{Proof of Theorem \ref{thm:clhardnesssuperres}:~}The Hamiltonian $H_{2D}[\vec{w}]$ is local and admits a diagonalization of the form (\ref{eq:specialHam}),  since it is diagonalized by the matrix $U_{2D}$ and its eigenvalues can be efficiently computed classically via the function $f(z)=\sum_i w_i z_i$, where $z_i$ denotes the bit decomposition of the integer $z\in\{0,2^n-1\}$. We will consider the case where $f(z)$ is an invertible function such that the value of $z$ can be efficiently computed from the energy value $f(z)$. In this case the probability of observing a particular energy value (up to error $\delta$) can directly be related to a certain output probability of the quantum circuit $U_{2D}$. For simplicity, we will consider the choice of weights $u_j=2^{-j}$, in which case $f(z)\equiv Id(z)= z 2^{-n}$, for $z\in\{0,...,2^n-1\}$ is a rescaled identity function, which is clearly invertible. Other choices could be possible without necessarily requiring exponentially decaying weights \footnote{Another possible choice of weights which leads to an invertible function $f(z)$ is $u_i=\log(p_i)/C$, where the $p_i$'s are different prime numbers and $C$ is a large enough positive number that guarantees the maximum energy to be less or equal to one. In this case $f(z)=\log(\prod_i p_i^{z_i})$ and given that each number has an unique prime decomposition and the values of $p_i$ are known, $z$ can be computed efficiently from the value of $f(z)$.}.

Let us consider a $\beta$-approximate energy sampler for the Hamiltonian $H_{2D}[\vec{u}]$, given by 
\begin{equation}\label{HId}
H_{2D}[\vec{u}]= \sum_j 2^{-j} U_{2D}^{\dagger} \hat{n}_j U_{2D}.
\end{equation}
To show how an efficient classical sampler from $U_{2D}$ can be constructed from a classical algorithm for super-resolution energy measurements, we consider the following sampling algorithm
\\~\\
\textbf{Algorithm 1}
\begin{itemize}
\item Query a $\beta-$approximate energy sampler with input Hamiltonian $H_{2D}[\vec{u}]$, input state $\rho=\ket{\psi_{\theta,x}}\bra{\psi_{\theta,x}}$ and parameters $\delta =2^{-n}/3$, $\eta=1-\epsilon$.

\item  Given an output $E_m=m\delta $, output the unique $z$ satisfying $E_m\in\{z/2^n-\delta ,z/2^n,z/2^n+\delta \}$.
\end{itemize}
The algorithm queries an energy sampler with a resolution one third smaller than the separation between consecutive eigenvalues of $\HamiltonianFamily{\vec{u}}$, which is $2^{-n}$, in order to guarantee that
the spectral projection in the interval $[z 2^{-n}-\delta, z 2^{-n}+\delta]$ is simply given by $\Pi_{z 2^{-n}}$. As shown below, this allows us to relate the probability that Algorithm 1 outputs $z$ to the quantity $P_z=|\bra{z}U\ket{\psi_{\theta,x}}|^2$,  the latter being the probability of observing the computational basis state $\ket{z}$ from a sampler from $U_{2D}$ with initial state $\ket{\psi_{\theta,x}}$.

To demonstrate this relation, we first analyze the case when $\beta=0$. Using the definition of measurement resolution from Eq.~\eqref{eq:Ensamp}, we obtain that the probability that Algorithm 1 outputs $z$ is given by 
\begin{align}
p_z&= \Pr(|E_m-z/2^n|\leq \delta )\\
&\geq (1-\epsilon)\bra{\psi_{\theta,x}}\Pi_{z 2^{-n}}\ket{\psi_{\theta,x}}\\
&\geq (1-\epsilon)|\bra{z}U_{2D}\ket{\psi_{\theta,x}}|^2\label{eq:ensampvsUsamp}
\end{align} 
The equality follows from the constraints on the probability distribution of an energy sampler, defined by Eq.~\eqref{eq:Ensamp}.
The first inequality results from the spectral projection in the interval $[z 2^{-n}-\delta, z 2^{-n}+\delta]$ being simply given by $\Pi_{z 2^{-n}}$. 
The second inequality follows from the fact that, by construction of $\HamiltonianFamily{\vec{u}}$, we have that
\begin{equation}
\Pi_{z 2^{-n}}=U^{\dagger}_{2D}\ket{z}\bra{z}U_{2D}.
\end{equation} 
From Eq.~\eqref{eq:ensampvsUsamp} we can define $dp_z\geq 0$ such that $p_z=(1-\epsilon)P_z+dp_z$. In addition, it can be seen, from the construction of Algorithm~1, that the probabilities $p_z$ are normalized, which implies that 
\begin{align}
&\sum_z p_z=1\Leftrightarrow \sum_z (1-\epsilon)P_z+dp_z=1\\
&\Rightarrow \sum_z dp_z=\epsilon.
\end{align}
Hence, it follows that 
\begin{align}
\sum_z |P_z-p_z|&=\sum_z |\epsilon P_z -dp_z| \nonumber\\
&\leq\epsilon+\sum_z dp_z
\leq 2\epsilon, \label{probbound2}
\end{align}
which shows that Algorithm~1 is an $2\epsilon$-approximate sampler for $U_{2D}$.

For the case where $\beta\geq 0$, Algorithm~1 outputs $z$ with probability $p'_z$
\begin{equation}\label{probbound3}
\sum_z|p'_z-p_z|\leq \beta.
\end{equation}
Consequently, the bounds from Eqs.~\eqref{probbound2} and \eqref{probbound3} imply that 
\begin{equation}
\sum_z |p'_z-P_z|\leq \beta+2\epsilon.
\end{equation}
This implies that Algorithm~1 is an efficient $\gamma$-approximate sampler for unitary $U_{2D}$ and initial state $\ket{\psi_{\theta,x}}$, with $\gamma=2\epsilon+\beta$.
\QEDB 

\begin{proof}[Proof of corollary \ref{cor:QASuperresolved}]
The classical hardness result for this energy measurement problem follows directly from the application of Theorem \ref{thm:clhardnesssuperres}, which maps this problem to that of sampling from $U_{2D}$, together with the quantum advantage result of Ref.~\cite{bermejo}. Specifically, Theorem~1 from Ref.~\cite{bermejo} states that a classical algorithm cannot approximately sample from outputs of the circuit $U_{2D}$, up to a total variation distance of $\beta=1/22$, in polynomial time. This result is based on three complexity theoretic assumptions which we summarize in section \ref{sec:conjectures}. 
\end{proof}
We remark that a crucial point of the proof of Theorem 2 is to choose the weights $w_i$ so the spectrum of the Hamiltonian is non-degenerate. It can be seen that this requires the weights $w_i$ to be defined up to $O(n)$ bits of precision. Such precision in controlling the parameters of a Hamiltonian is hard to achieve in an experimental setting. Nevertheless, measuring the energy of quantum states according to such Hamiltonians is a valid theoretical question that can be explored using quantum devices beyond what is likely to be simulable efficiently classically given Corollary~\ref{cor:QASuperresolved} . We leave as an open problem whether a similar hardness result to Corollary~\ref{cor:QASuperresolved} can be achieved for measurements of Hamiltonians whose parameters are defined up to constant precision.
\subsection{Complexity theoretic assumptions needed for quantum advantage via sampling problems}\label{sec:conjectures}
The quantum advantage proposal of Corollary \ref{cor:QASuperresolved} relies on complexity-theoretic conjectures. Specifically,  these are the same conjectures\footnote{This follows, from the proofs of Theorem \ref{thm:clhardnesssuperres} and Corollary \ref{cor:QASuperresolved}, since we exploit a reduction to the quantum sampling problem therein.} underlying  the quantum advantage proposal of Ref.~\cite{bermejo}; the latter are, in turn, analogous to those involved in hardness proofs for random universal quantum circuits \cite{boixo_characterizing_2016,bouland_quantum_2018} and slightly weaker than those in the seminal boson sampling and IQP proposals of Refs.\mot\cite{bosonsampling,IQPapprox}. Any progress towards proving the conjectures in Ref.~\cite{bermejo} would simultaneously improve our results.

For the sake of completeness we here summarize the conjectures that enter the quantum advantage results in \cite{bermejo}. The first (C1) is the non-collapse of the Polynomial Hierarchy, a widely believed generalization of the  $\textsf{P}\neq\textsf{NP}$ conjecture \cite{karp1980,fortnow2005beyond,aaronson2016},  first linked to hardness of classical simulation of quantum circuits in \cite{constantdepth}. This assumption alone rules out the existence of near-exact classical simulators for a large family of quantum devices, including the ones we study: specifically, this is the case for quantum devices with output probabilities that are \#P-hard to approximate up to constant relative errors in worst case (see \cite{bosonsampling,IQPexact} and appendix\mot\ref{app:Stockmeyer}). 

Following an approach pioneered in \cite{bosonsampling,IQPapprox}, this classical hardness result can be made noise-robust up to a constant sampling error in total variation distance\mot$\beta$ assuming two additional conjectures about the output probabilities of the quantum device. Specifically, let the set of all output probabilities of our device be $\mathcal{P}_n=\{p_x,x\in\{0,1\}^n\}$, where $n$ is the number of qubits. Then, we require: (C2) an approximate average-case hardness conjecture, which states that a constant fraction of the output probabilities in $\mathcal{P}_n$ are \#P-hard to approximate (up to a constant relative error); (C3)  an anticoncentration conjecture, stating that the output distribution is ``sufficiently flat'', in the sense that $\mathrm{prob}_{p_x\sim \mathcal{P}_n } \left[p_x > \alpha/2^n\right]> \gamma$, for some constants $\alpha, \gamma\in O(1)$. Assuming (C2)-(C3), it can be shown that the existence of an efficient classical algorithm for $\beta$-approximate sampling from the unitary $U_{2D}$ implies the collapse of the Polynomial Hierarchy to its 3rd level. The central technique in this argument is Stockmeyer's counting algorithm \cite{stockmeyer1985approximation} (Sec.~\ref{sec:Stockmeyer} and Appendix  \ref{app:Stockmeyer}), which shows that such a classical sampler implies the existence of an algorithm 
(inside the third level of the Polynomial Hierarchy) for estimating the probabilities in $\mathcal{P}_n$ on average with high accuracy, a \#P-hard problem. This then implies the aforementioned collapse of complexity classes by Toda's theorem\cite{Toda:1991:PHP:125944.125952} (see appendix\mot\ref{app:Stockmeyer} for more details). We note that the specific total variation distance $\beta$ tolerated for the classical sampler depends on the choices of the constant parameters in the statement of the conjectures; the values chosen in Ref.~\cite{bermejo} lead to the threshold $\beta=1/22$.

There has been steady progress towards proving the complexity theoretic assumptions made. In Ref.~\cite{bermejo}, numerical evidence was presented which supports the anticoncentration conjecture (C3), for the choice of angles of the input state from the set $\Theta=\{0,\pi/4\}$. Moreover, this conjecture was recently proven for a uniform random choice of angles of the input state (see Eq.~\ref{eq:randomstates}) from the set $\Theta=[0,2\pi]$ \cite{haferkamp2019closing}. Furthermore, positive evidence for conjecture (C2) is given by approximate worst-case hardness results in \cite{bermejo} (case $\Theta=\{0,\pi/4\}$), as well as by recent proofs of exact average-case hardness and anticoncentration theorems given in \cite{haferkamp2019closing} for the larger set of angles (case $\Theta=[0,2\pi]$). The results in \cite{bermejo,haferkamp2019closing} are analogous to approximate worst-case hardness results in  \cite{bosonsampling,IQPapprox,boixo_characterizing_2016}; proofs of anti-concentration of the output distribution in \cite{IQPapprox,Hangleiter2018anticoncentration}; and exact average-case hardness results in \mot\cite{Aaronson-ProcRS-2005,bouland_quantum_2018}. For all known quantum advantage proposals, including ours, proving approximate average-case hardness remains an open question.
%
\section{Near-exact simulation of energy measurements with standard resolution}\label{sec:standardres_exact}
The previous result establishes a complexity-theoretic obstruction for classical algorithms to simulate  energy measurements of certain local Hamiltonians in the super-resolution regime. In this section, we investigate whether the less restrictive standard resolution regime can be efficiently reached classically. Our main result shows this problem remains hard, with high complexity theoretic evidence (specifically, a collapse of the Polynomial Hierarchy, introduced in section \ref{sec:clhardness}), if we demand exponentially small error parameters.
\begin{thm}[\textbf{Hardness of near-exact standard-resolution energy measurements}]\label{thm:HardnessInaccurateMeasurement}
Let $H$ be a local Hamiltonian acting on $n$ qubits. If there exists an efficient classical algorithm for the energy sampling problem for any such $H$ with product-state inputs; resolution $\delta \in O(1/\poly(n))$; confidence $\eta = 1-O(1/2^{\poly(n)})$; and sampling error $\beta\in O(1/2^{\poly(n)})$; then the Polynomial Hierarchy collapses to the 3rd level. Furthermore, the same holds if $H$ is a nearest-neighbor, translation invariant 5-local Hamiltonian on a 2d lattice.
\end{thm}
To demonstrate this theorem, the idea is to show that such a classical algorithm would efficiently sample an outcome whose probability is $\#$P-hard to approximate. 
This is achieved by constructing local Hamiltonians with two properties:
\begin{itemize}
    \item[(i)] there is a unique ground state with energy 0, and a polynomially small gap to the first excited state.
    \item[(ii)] the overlap of this ground state with a product state is $\#$P-hard to calculate up to an inverse-exponential additive error. 
\end{itemize}
The first condition ensures that a standard-resolution energy measurement protocol is capable of efficiently discriminating the ground state from the rest of the spectra, while the second ensures that a near-exact sampling measurement samples from a $\#$P-hard probability.

We provide two examples of families of local Hamiltonians that fulfill these properties in Secs.~\ref{sec:physicalexamp}. The first example is presented in Secs.~\ref{sec:TIHam} and is a  translation-invariant version of the 5-local  cluster state Hamiltonian from Eq.~\eqref{eq:H2dlocal}. By construction, the ground state is related to an output state of the quantum circuit $U_{2D}$ and hence property (ii) follows directly from the results of Ref.~\cite{bermejo}.
The second example, presented in Sec.~\ref{sec:FKhamiltonians}, is a 4-local Hamiltonian based on Feynman-Kitaev (FK) circuit-to-Hamiltonian constructions \cite{kitaev2002,aharonov2008adiabatic}, used in the proof of equivalence between adiabatic and circuit model quantum computation \cite{aharonov2008adiabatic}. This construction, although more complicated than the first example, gives a general technique to relate the output state of an arbitrary poly-size quantum circuit to the ground state of a local Hamiltonian. Consequently, all the results on $\#$P-hardness of output probabilities of quantum circuits can be translated to results on hardness of Energy Sampling with standard resolution.

Using these examples of local Hamiltonians, we present the proof of Theorem~\ref{thm:HardnessInaccurateMeasurement} in Sec.~\ref{sec:proofthm3}, following standard arguments from the quantum advantage literature. 
\subsection{Physical examples}\label{sec:physicalexamp}
\subsubsection{A 5-local translationally invariant Hamiltonian}\label{sec:TIHam}
We consider the Hamiltonian family defined in Eq.~\eqref{eq:H2dlocal}. If we pick a uniform choice of weights $v_j:=1/n,j=1,\ldots,n$, the resulting Hamiltonian is 5-local nearest-neighbor and translation-invariant. In particular, up to single-qubit rotations, it is the well-known 2D cluster state Hamiltonian \cite{Briegel-PRL-2001,RaussendorfPhysRevLett.86.5188}:
\begin{equation}\label{eq:H2dunif}
\HamiltonianFamily{\vec{v}}=\frac{1}{n}\sum_{j=1}^n X_j \prod_{j:(j,l)\in E_{2D}} Z_l.
\end{equation}
This model is mapped to a trivial unentangled one via the unitary (\ref{U_2D}), which has constant depth.  In a condensed matter sense, this implies that $\HamiltonianFamily{\vec{v}}$ can be seen as the energy density operator (the Hamiltonian divided by the number of particles in the lattice) of a gapped model in the trivial phase \cite{sachdev11,YichenChen15}. In particular, $\HamiltonianFamily{\vec{v}}$ has an inverse polynomial gap $\Omega(1/n)$, which ensures that a standard-resolution measurement can efficiently discriminate its ground state from the rest of the spectra.

Upon an ideal energy measurement of a state $\ket{\psi_{\theta,x}}$ from Eq.~\eqref{eq:randomstates}, the probability of obtaining outcome $E=0$, which is the ground-state energy, is given by 
\begin{equation}\label{eq:p_gs1}
P_{GS}=|\bra{0}^{\otimes n} U_{2D} \ket{\psi_{\theta,x}}|^2.
\end{equation} 
Such probabilities are related to partition functions of Ising models and were shown to be $\#$P-hard to approximate to relative error or to inverse-exponential additive error \cite{IQPapprox, bermejo}. This is a crucial ingredient for the proof of Theorem~\ref{thm:HardnessInaccurateMeasurement}, presented in Sec.~\ref{sec:proofthm3}. 
\subsubsection{ Circuit-to-Hamiltonian constructions}\label{sec:FKhamiltonians}
So far, we have only considered Hamiltonians whose diagonalization is known. Here we present a more general strategy to relate probabilities of outputs of \emph{arbitrary quantum circuits} to probabilities of outcomes of energy measurements, while keeping the Hamiltonian local. To do so, we use the so-called circuit-to-Hamiltonian constructions, based on Feynman clocks, that have been widely used in Hamiltonian complexity and adiabatic quantum computation \cite{kitaev2002, aharonov2008adiabatic}.
In fact, the Hamiltonian that has the properties we require is used in the proof that the adiabatic model of quantum computation is equivalent to the circuit model \cite{aharonov2008adiabatic}. More precisely, we consider the final Hamiltonian at the end of the adiabatic schedule, since it has polynomially small gap and its ground state contains a information about the final state of a quantum computation, as we explain in what follows. We describe this construction in more detail since we build upon it to demonstrate our main result in Sec.~\ref{sec:polybox_enmeas}.

Let us consider a quantum circuit of $T=\poly(n)$ gates $U=U_T U_{T-1}...U_1$ and the propagation Hamiltonian
\begin{multline}\label{eq:Hprop}
H_{prop}=\frac{1}{2}\sum_{t=1}^T\mathbbm{1}\otimes\ket{t}\bra{t}_c+\mathbbm{1}\otimes\ket{t-1}\bra{t-1}_c\\- U_t\otimes \ket{t}\bra{t-1}_c - U_t^{\dagger}\otimes \ket{t-1}\bra{t}_c,
\end{multline}
which is defined on a Hilbert space with $T+1$ clock states $\ket{t}_c$ and $n$ qubits. We will describe the case, where the clock is implemented with $O(\log(T+1))$ qubits, in which case the Hamiltonian $H_{prop.}$ is $O(\log(n))$-local. Nevertheless our result extends to 5-local Hamiltonians using the unary clock implementation of Ref.~\cite{kitaev2002}, or 4-local Hamiltonians using a clock implementation based on the hopping of a excitation in a unidimensional spin-chain (see \cite{caha2018} for a recent discussion on different clock implementations). We define the states 
\begin{align} 
\ket{\eta_y(0)}&=\ket{y}\ket{0}_c\nonumber\\
\ket{\eta_y(t)}&= U_t U_{t-1}...U_1\ket{y}\ket{t}_c,  t\in \{1, ..., T\},
\label{eq:stateseta}
\end{align}
and the subspaces 
\begin{equation}\label{eq:subspaces}
\Omega^{(y)}=\text{span}\{\ket{\eta_y(t)}, t= 0, ..., T\}.
\end{equation}
We note that for each $y$ the subspace $\Omega^{(y)}$ is invariant under the action of the Hamiltonian $H_{prop}$. Hence, we can diagonalize the Hamiltonian in each of these subspaces, obtaining the $2^n$ degenerate ground states 
\begin{equation}\label{eq:historystate}
\ket{\psi^{(y)}}=\frac{1}{\sqrt{T+1}}\sum_{t=0}^T \ket{\eta_y(t)}
\end{equation} 
which have energy $0$. These are called history states as they contain information about the whole history of a quantum computation. In particular, we have that $\bra{x}\bra{t}_c\ket{\psi^{(y)}} = \frac{1}{\sqrt{T+1}} \bra{x}U_t...U_1 \ket{y}$, which is proportional to the transition amplitude from state $y$ to state $x$ after $t$ steps of the computation. It can be shown that $H_{prop}$ is positive semidefinite and has a gap of $O(1/T^2)$ with respect to the first excited state~\cite{kitaev2002}. 

To fix the initial state of the computation to be $\ket{0}$ it is necessary to add an energy penalty Hamiltonian of the form 
\begin{equation}
    H_{init}= \sum_{i=1}^n\ket{1}\bra{1}_i\otimes \ket{0}\bra{0}_c,
\end{equation}
where the projector $\ket{1}\bra{1}_i$ acts on the $i$\emph{th} qubit, ensuring that when the clock is in its initial state $\ket{0}_c$, any computational basis state which is not the $\ket{0}$ state is energetically penalized. The total Hamiltonian 
\begin{equation}\label{eq:HamFK}
    H=H_{init}+H_{prop}
\end{equation}
has a single ground state $\ket{\psi^{(0)}}$ with energy $0$ and gap $\Delta=O(1/T^3)$ \cite{kitaev2002}, ensuring the discernibility of the groundstate via a standard-resolution measurement. Hence, the probability of observing the ground state of $H$ upon an ideal energy measurement of a state $\ket{y}\ket{T}_c$ is given by 
\begin{align}
    P_{GS}&= |\bra{\psi^{(0)}}(\ket{y}\ket{T}_c)|^2\\
    &=\frac{1}{T+1} |\braket{y|U|0}|^2\label{eq:PGSFKHam}
\end{align}
This quantity is $\#$P-hard to estimate to relative error or inverse-exponential additive error for several families of quantum circuits such as IQP~\cite{IQPapprox, ni2012commuting}, or boson sampling \cite{bosonsampling} (which can also be implemented in the circuit model \cite{Qcircuitbosonsampling}), among others. Depending on the family of circuits we choose, this defines a family of local Hamiltonians of the form given by Eq.~\eqref{eq:HamFK} for which Theorem~\ref{thm:HardnessInaccurateMeasurement} applies. 
\subsection{Proof of Theorem \ref{thm:HardnessInaccurateMeasurement}}\label{sec:proofthm3}
Before we proceed, let us prove the following technical lemma that will be useful in what follows. 
\begin{lem}\label{lem:probGS}
Let $H$ be a Hamiltonian with eigenvalues in the interval $[0,1]$, a ground state energy $E_{GS}=0$ and a gap $\Delta$ to the first excited state. Given this Hamiltonian and an initial state $\ket{\psi}$, consider a $\beta-$approximate Energy sampler with $\delta=\Delta/3$ and confidence $\eta=1-\epsilon$. Let $q'_{GS}$ be the probability that this sampler outputs a value $E_m$ in the interval $[0,\Delta/3]$.
We have the following bounds on $q'_{GS}$: 
\begin{equation}\label{eq:bounds}
|q'_{GS}-P_{GS}| \leq \epsilon+\beta. 
\end{equation}
where 
\begin{equation}
P_{GS}=\braket{\psi|\Pi_{GS}|\psi}
\end{equation}
and $\Pi_{GS}$ is the spectral projection in the ground state space of $H$.
\label{thm:GSprob}
\end{lem}

\begin{proof} 
An energy sampler with parameters $\eta=1-\epsilon$ and $\delta=\Delta/3$ outputs a value in the interval $[0,\Delta/3]$ with probability 
\begin{equation}\label{eq:lowbound}
q_{GS}\geq (1-\epsilon) P_{GS}\geq P_{GS}-\epsilon,
\end{equation}
which follows from the constraints on the outcome probability distribution defined by Eq.~\eqref{eq:Ensamp}. On the other hand, the probability of obtaining a value in the interval $[2\Delta/3,1]$ obeys the bound 
\begin{align}
\bar{q}&\geq (1-\epsilon) \bra{\psi}\Pi_{[\Delta,1]}\ket{\psi}\\
&=(1-\epsilon) (1-P_{GS}) \label{eq:boundqbar},
\end{align}
where we used Eq.~\eqref{eq:Ensamp} in the first step and the fact that $\Pi_{GS}+\Pi_{[\Delta,1]}=I$ in the second step, which follows from the assumption that the Hamiltonian has a gap $\Delta$. In addition, the $\bar{q}$ and $q_{GS}$ probabilities sum up to 1, which implies that 
\begin{align}
q_{GS}&=1-\bar{q}\\
&\leq 1- (1-\epsilon)(1-P_{GS})\\
&= (1-\epsilon) P_{GS}+\epsilon\leq P_{GS}+\epsilon,\label{eq:upbound}
\end{align}
where in the second step we used Eq.~\eqref{eq:boundqbar}. Combining equations (\ref{eq:lowbound} and \eqref{eq:upbound}) we obtain the inequality
\begin{equation}
|q_{GS}-P_{GS}| \leq \epsilon.
\label{eq:prebound}
\end{equation}

A $\beta$-approximate Energy sampler with parameters $\eta=1-\epsilon$ and $\delta=\Delta/3$ outputs a value in the interval $[0,\Delta/3]$ with probability $q'_{GS}$, where $|q'_{GS}-q_{GS}|\leq \beta$. Hence, using \eqref{eq:prebound} and the triangle inequality we conclude the proof. 
\end{proof}

 With this lemma we are ready to prove Theorem \ref{thm:HardnessInaccurateMeasurement}.
\begin{proof}[Proof of theorem \ref{thm:HardnessInaccurateMeasurement}]
Lemma~\ref{thm:GSprob} implies that an approximate energy sampler with $\delta =\Delta/3=1/(3n+3)$, $\epsilon=1-\eta=1/2^{\poly(n)}$ and $ \beta=1/2^{\poly(n)}$ would output $E_m\in [0,\Delta/3]$ with a probability $q_{GS}=P_{GS}+\epsilon'$, where $|\epsilon'|$ is an exponentially small number. Furthermore, we have seen two constructions of local Hamiltonians for which $P_{GS}$ is $\#$P-hard to estimate with inverse-exponential additive error (Eqs.~\eqref{eq:p_gs1} and \eqref{eq:PGSFKHam}).
Let us assume there is an efficient classical energy sampler with the parameters defined in Theorem~\ref{thm:HardnessInaccurateMeasurement} for these Hamiltonians. Following standard arguments in the literature of quantum advantage \cite{bosonsampling, IQPapprox}, this would imply the probability $q'_{GS}$ could be estimated up to an inverse-exponential additive error via Stockmeyer's algorithm, an algorithm in the third level of the polynomial hierarchy (PH). This implies that a $\#$P-hard problem could be solved in the third level of the PH and hence the PH would collapse to the third level. Stockmeyer's algorithm and its connection to quantum advantage is reviewed in more detail in Sec.~\ref{sec:Stockmeyer} and Appendix~\ref{app:Stockmeyer}.
\end{proof}
This gives strong evidence to the impossibility for classical computers to efficiently simulate energy sampling problems with confidence exponentially close to optimal, i.e., $\eta=1-1/2^{\poly(n)}$, inverse-exponential $\beta=1/2^{\poly(n)}$, and standard resolution $\delta=1/\poly(n)$. This can be seen as a classical hardness result for the problem of simulating an ideal implementation of the quantum phase estimation algorithm (with confidence amplification~\cite{aharonov2017}), for measuring the energy of a local Hamiltonian with standard resolution.
\section{Computational Complexity of approximate energy sampling with standard resolution}\label{sec:standardres_approx} 
The previous section provided evidence that classical algorithms for near-exact simulations of energy measurements of many-body Hamiltonians cannot efficiently simulate measurements with standard resolution ($\delta=1/\poly(n)$). From a physical perspective, it is natural to ask whether the previous complexity theoretic results involving collapses of the Polynomial Hierarchy can be extended to approximate simulations.  Specifically, we are interested in extending Theorem \ref{thm:HardnessInaccurateMeasurement} to the regime where the measurement failure probability $\epsilon$ and the sampling error $\beta$ are small constants. We present a no-go lemma and one positive result.  

Approximate sampling measurements of Hamiltonians in the standard resolution regime can be interpreted as examples of quantum sampling problems with a small number of output qubits. It would thus be tempting to apply the Stockmeyer-based techniques of Refs.\mot\cite{IQPapprox,bosonsampling} (cf. section \ref{sec:conjectures}) to study the complexity theory of classically simulating such measurements. Unfortunately, we first point out in Lemma \ref{lemma:ObstructionBMS} (section \ref{sec:Stockmeyer}) that the Stockmeyer argument cannot meaningfully link the hardness of approximate sampling problems with few outputs to a Polynomial Hierarchy collapse.  This is due to an error parameter in such proofs that becomes too large precisely for quantum computations where the number of measured output qubits is ``small'': constant or $O(\log(n))$. This issue is generic and affects, e.g., existing quantum advantage proposals based on variations of the one-clean-qubit (DQC1) model\mot\cite{Morimae14HardnessDQC1,Morimae1704.03640}.

In spite of the above hurdle, our second result (Theorem \ref{thm:polybox}, section \ref{sec:polybox_enmeas} below) gives complexity theoretic evidence of the classical hardness of approximate standard-resolution energy sampling. The result links the worst-case complexity of this problem to that of classically simulating universal quantum computers. Specifically, we prove that the existence of an efficient classical algorithm for this problem would imply the ability to  efficiently classically compute arbitrary marginal output probabilities of  universal poly-sized quantum circuits (a BQP-hard task, in the language of complexity theory\mot\cite{nielsen2002quantum}). This provides evidence against an efficient classical simulation of energy measurements with standard resolution.

The first result in this section highlights the existence of a gap in the complexity theoretic understanding of quantum approximate sampling problems with small output support. The second opens the possibility to develop quantum advantage tests based on such problems. This is however complicated by the lack of tools to study the average-case hardness of this problem. We discuss this latter possibility and associated open challenges in section \ref{sec:OpenQuestions}.
\subsection{Sampling problems with small support ``do not simply'' collapse the Polynomial Hierarchy}\label{sec:Stockmeyer}
Here, we point out a technical obstruction towards extending available approaches in quantum advantage proofs \cite{IQPapprox,bosonsampling} to the standard-resolution approximate energy sampling problems. First, we remark that any algorithm for energy measurements with standard resolution samples from a probability distribution with $\poly(n)$ outcomes. This is in contrast with most quantum advantage proposals, which have an outcome space that is exponentially large. This fact constitutes a roadblock for the application of the proof technique of  Refs.~\cite{IQPapprox,bosonsampling}.

To understand the limitation, we recall (section \ref{sec:conjectures}) that the traditional approach to prove quantum advantage results via sampling problems heavily relies on  Stockmeyer's algorithm Refs.~\cite{IQPapprox,bosonsampling}. The goal there is to induce a Polynomial Hierarchy (PH) collapse assuming, among other assumptions, that it is \#P-hard  to approximate the output probabilities of a quantum device up to very small errors: specifically, a constant relative one if we have anticoncentration. Unfortunately, as shown next, if we tried to adapt the same argument to rule out classical algorithms for sampling problems with poly-sized support, we would have to adopt an analogous average-case conjecture where the error is too large for the assumption to be plausible. Below, we characterize these errors for circuits with an arbitrary number of output bits. Let $q_U$ be the output probability distribution of a quantum circuit $U$, and $0_m$ be the string with $m$  zeroes.

\begin{lem}[\textbf{Stockmeyer error}]\label{lemma:ObstructionBMS}
Let $\mathcal{Q}_n,{n\in\mathbb{N}}$ be a family of uniformly-generated poly-size $n$-qubit quantum circuits  with $m$ output bits and the hiding property
$$\forall U \in \mathcal{Q}_n, x\in\{0,1\}^m, \exists  U_x\in \ \mathcal{Q}_n : q_U(x)=q_{U_x}(0_m).$$
Assume there exists a classical algorithm $\mathcal{A}$ that samples from $q_U$ with $\ell_1$ error $\beta$ in $O(\poly(n,1/\beta))$ time given $U\in\mathcal{Q}_n$. Then, for any $0<\nu<1$, there is an $\textnormal{FBPP}^\textnormal{NP}$ algorithm which, given access to $\mathcal{A}$, approximates $q_U(x),x\in\{0,1\}^m$ up to additive error $\varepsilon$
\begin{equation}
\varepsilon\in O \left( \frac{q_U(x)}{\poly(n)}+ \frac{\beta}{2^m\nu}\left(1+\frac{1}{\poly(n)} \right)\right).\label{eq:Stockmeyer}
\end{equation}
with probability $1-\nu$ over the choice of~$x\in\{0,1\}^m$. 
\end{lem}
We provide the proof of this lemma in Appendix \ref{app:Stockmeyer}. There, we also argue in detail how the Stockmeyer argument fails to provide a plausible collapse of PH in the $m<\log(n)$ regime. The basic intuition is as follows. In Refs. \cite{bosonsampling,IQPapprox}, where $m=n$, the algorithm provides a relative error estimation of the output probabilities in average case if we assume anticoncentration. This problem is then conjectured to be $\#$P-hard. Evidence for this conjecture is provided by worst-case results and near-exact worst-to-average reductions (section \ref{sec:conjectures}). Because of Toda's theorem, that provides a collapse of PH, since PH$\subset$P$^\textrm{\#\textrm{P}}$.
In the $m<\log(n)$ regime, the right hand side of Eq. (\ref{eq:Stockmeyer}) has  a term that can only be upper bounded by an inverse polynomial, which limits the accuracy $\varepsilon$ of the algorithm that estimates probabilities in $\textnormal{FBPP}^\textnormal{NP}$. Unfortunately, to induce the same collapse of PH in the case $m<\log(n)$, we would need to show that it is $\#$P-hard to estimate quantum output probabilities with an inverse polynomial error. This is however implausible because, if  it was true, then quantum computers could efficiently solve $\#$P-hard problems, which is believed to be impossible~\cite{Bennet97StrenghtsWeaknnesses_QC,Aaronson-ProcRS-2005}.
\subsection{Hardness of approximate energy measurements with standard resolution}\label{sec:polybox_enmeas}
In the previous section, we discussed obstructions towards proving quantum advantage results based on known complexity theoretic conjectures for approximate standard-resolution energy sampling, as well as quantum sampling problems with small support. This points towards a tension between having practical physically-motivated quantum advantage schemes and strong complexity theoretic proofs of classical hardness. 

This motivates us to consider  different approaches to prove physically-motivated quantum advantage results, which do not rely on the Stockmeyer argument of Refs.\mot\cite{IQPapprox,bosonsampling}. 
In fact, alternative evidence for the impossibility of developing efficient classical algorithms to simulate approximate energy measurements with standard resolution can be drawn from the work of Refs.~\cite{wocjan2006BQP, janzing2008BQP}. Therein, the authors show that a procedure for energy measurements of local Hamiltonians achieving a resolution of $1/\poly(n)$ and confidence $\eta=1-\epsilon$, where $\epsilon$ is a small constant, can be used to decide any problem in BQP, the class of decision problems that can be efficiently solved by a quantum computer\mot\cite{nielsen2002quantum}. The Hamiltonians considered therein are $4$-local non-nearest neighbor Hamiltonians in \cite{wocjan2006BQP} and translational invariant chains of qudits in \cite{janzing2008BQP}. Although these works did not explicitly consider $\beta$ sampling errors, they can be easily  extended to the approximate energy sampling regime where $\beta$ is a small constant.
Consequently, the existence of a classical algorithm for standard resolution approximate energy sampling problems, would imply that classical computers could solve efficiently any problem in BQP.

In this section, we provide additional complexity-theoretic evidence that classical computers cannot efficiently simulate energy measurements with standard resolution. We do so by showing that this problem is at least as hard as estimating marginals of output probability distributions of universal circuits, a problem that is more general than considering only decision problems solvable by quantum circuits. Specifically, our main result (Theorem \ref{thm:polybox}) shows that the ability to simulate the approximate energy sampling problem efficiently would imply the existence of a ``poly box'', in the notation of \mot\cite{pashayan2017}: i.e., an efficient algorithm to estimate any marginal output probability of any poly-size quantum circuit up to a polynomially small error,
a BQP-hard task.
\begin{definition}[Probability estimator or  ``poly-box'']\label{def:polybox} Let $U$ be a poly-size quantum circuit acting on $n$-qubits. An algorithm is said to be a probability estimator or poly-box for $U$ if it can compute an estimate $\hat{p}$ of any marginal probability $p$ of the distribution $|\bra{x}U\ket{0}|^2$ such that
\begin{equation}\label{eq:defpolybox}
\Pr(|p-\hat{p}|\leq \delta_p)\geq 1- \epsilon_p
\end{equation}
in time $O(poly(n, \delta_p^{-1}, \log({\epsilon_p}^{-1})))$.
\end{definition}
The connection between standard-resolution energy measurements and probability estimators is stated precisely in the following theorem.

\begin{thm}[\bf Hardness of approximate standard-resolution energy measurements]\label{thm:polybox}
Let us assume the existence of a classical algorithm for approximate energy sampling for 4-local Hamiltonians on product states, reaching a resolution $\delta$, confidence $\eta=1-\epsilon$ and sampling error $\beta$,
 with a running time of $O(\poly(n,\delta^{-1},\beta^{-1},\epsilon^{-1}))$. This implies existence of a classical poly-box for arbitrary poly-size quantum circuits.
\end{thm}
Theorem\mot\ref{thm:polybox} shows that standard-resolution energy measurements can be used to estimate arbitrary output probabilities of quantum circuits, and not just of single-qubit measurements,  generalizing the results of~\cite{wocjan2006BQP,janzing2008BQP}.

To prove theorem \ref{thm:polybox}, we show that it is possible to encode any marginal probability of a quantum circuit's output distribution as the probability of measuring the ground state energy of a certain 4-local Feynman-Kitaev Hamiltonian, which has a polynomially small gap. Hence, with a polynomial number of energy measurements, the marginal probability can be estimated with a polynomially small error via the Hoeffding bound. 
 
\begin{proof}[Proof of Theorem \ref{thm:polybox}]
Let $p$ be an output probability or a marginal probability of a poly-size quantum circuit $U$ from a family of quantum circuits $\mathcal{C}$ acting on $n$-qubits. For a fixed computational basis state input of the circuit $\ket{x}$, we can write the marginal probability as
 \begin{equation}
p=\sum_{y\in S^*} |\bra{y}U\ket{x}|^2.
\end{equation}  
for a given set of bit strings $S^*$. More precisely, we define $S^*$ has a set of $2^{n-l}$ bit strings of size $n$ where $l$ bits are fixed. We pick the bits at different positions $k_i$, with $k_i\in\{1,...,n\}$, $i\in\{1,...,l\}$, such that the $k_i$\emph{th} bit is fixed to a chosen value $b_i\in \{0,1\}$ i.e., 
\begin{equation}
S^*=\{y~: y_{k_i}=b_i, i\in \{1,..,l\}\}.
\end{equation}
To demonstrate the theorem we will first show that it is possible to construct a local Hamiltonian with two properties: the ground state energy is $E_{GS}=0$ and there is a polynomially small gap to the first excited state. Moreover, the probability of observing the outcome $E_{GS}=0$ after an energy measurement of this Hamiltonian on a product state is given by $P_{GS}=p/(T+1)$, where $T$ is the number of gates of circuit $U$. Such Hamiltonian can be constructed by a simple modification of the circuit-to-Hamiltonian construction mentioned in Sec.~\ref{sec:FKhamiltonians} in the following way.

Let us consider the gate decomposition of circuit $U=U_T U_{T-1}...U_1$ and the propagation Hamiltonian from Eq.~\eqref{eq:Hprop}. Similarly to Sec.~\ref{sec:FKhamiltonians}, we will prove our result for the simplest case, where the clock is implemented with $O(\log(T+1))$ qubits and the Hamiltonian $H_{prop.}$ is $O(\log(n))$-local. The physicality of the Hamiltonian can be improved to 4-local or 5-local, using standard clock implementations~\cite{kitaev2002, caha2018}.

We recall that the subspaces 
\begin{equation}\label{eq:subspaces2}
\Omega^{(y)}=\text{span}\{\ket{\eta_y(t)}, t= 0, ..., T\},
\end{equation}
with the states $\ket{ \eta_y(t)}$ defined in Eq.~\eqref{eq:stateseta},
are invariant under the action of the Hamiltonian $H_{prop}$. Furthermore, $H_{prop}$ has $2^n$ degenerate ground states $\ket{\psi^{(y)}}$ (see Eq.~\eqref{eq:historystate}) with energy zero. Another important property that will be used in the following proof is that $H_{prop}$ is positive semidefinite and has a gap of $O(1/T^2)$ with respect to the first excited state \cite{kitaev2002}. 

To relate the probability of observing the ground state to the marginal probability $p$ we need to lift the ground state such that $\ket{\psi^{(y)}}$ is a ground state only for $y\in S^{*}$. With this aim, we introduce the following penalty Hamiltonian
\begin{equation}
H_{pen}=\sum_{i=1}^l \left(\ket{\bar{b_i}}\bra{\bar{b_i}}\right)_{k_i}\otimes \ket{0}\bra{0}_c,
\end{equation}
where $\bar{b_i}$ denotes the NOT of the bit $b_i$ and the projector $\left(\ket{\bar{b_i}}\bra{\bar{b_i}}\right)_{k_i}$ 
acts non-trivially only on the $k_i$\emph{th} qubit (and as identity in the other qubits). It can easily be checked that $\ket{\eta_y(t)}$ are eigenstates of $H_{pen}$ with eigenenergies 
\begin{equation}\label{eq:enHpen}
\bra{\eta_y(t)}H_{pen}\ket{\eta_y(t)}=
\begin{cases}
0, \text{~~~if~} t\neq 0\\
0, \text{~~~if~} t= 0 \wedge y\in S^*\\
C_y, \text{~if~} t= 0 \wedge y\notin S^*\\
\end{cases}
\end{equation}
where $C_y\geq 1$, since if $y\notin S^*$ then at least one bit of $y$ in one of the positions $k_i$ is in state $\bar{b_i}$. From Eq.~\eqref{eq:enHpen} it can be seen that $H_{pen}$ has no effect in the subspaces $\Omega^{(y)}$ with $y\in S^*$. 

Let us now determine the ground states of $H=H_{prop}+H_{pen}$ as well as the gap to the first excited state. First, let us note that the subspaces $\Omega^{(y)}$ from Eq.~\eqref{eq:subspaces2} are also invariant under the action of $H_{pen}$, which trivially follows from the fact that $\ket{\eta_y(t)}$ are eigenstates of $H_{pen}$. Furthermore, since both $H_{prop}$ and $H_{pen}$ are positive semidefinite matrices, $H$ is also positive semidefinite.

Let us denote as $H^{(y)}$ and $H_{prop}^{(y)}$ the Hamiltonian $H$ and $H_{prop}$ restricted to the subspace $\Omega^{(y)}$, respectively. Then we have
\begin{align}
H^{(y)}&=H_{prop}^{(y)}, \text{~~~if~} y\in S^*\\
H^{(y)}&=H_{prop}^{(y)}+ C_y \ket{\eta_y(0)}\bra{\eta_y(0)}, \text{~~~if~} y\notin S^*,
\end{align} 
which follows from Eq.~\eqref{eq:enHpen}. Hence, for $y\in S^*$ the state of $H^{(y)}$ with the lowest energy is $\ket{\psi^{(y)}}$, which has energy $0$, and the first excited state has energy $O(1/T^2)$. 

The final step needed to demonstrate that these are the only ground states of $H$ is to show that the state with the lowest energy of $H^{(y)}$, with $y\notin S^*$, has an energy at least of $O(1/T^3)$. This implies that no state belonging to the subspace $\Omega^{(y)}$ with $y\notin S^*$ is a ground state of the whole Hamiltonian $H$ and that this Hamiltonian has indeed a gap of $1/\poly(n)$. To show this we use the geometrical lemma \cite{kitaev2002, aharonovsurvey}.
\begin{lem}\textbf(Geometrical Lemma)
Let $H_1$ and $H_2$ be two Hamiltonians with ground state energies $g_1$ and $g_2$, respectively. Also, let $\Delta_1$ and $\Delta_2$ be the their respective gaps to their first excited states. Then the ground state energy of $H$ is
$g\geq g_1+g_2+\Delta(1-\cos(\theta))$, where $\Delta=\text{min}(\Delta_1,\Delta_2)$ and $cos(\theta)$ is the maximum possible absolute value of the overlap between a ground state of $H_1$ with a ground state of $H_2$.   
\end{lem}  
We will use this lemma considering $H_1=H_{prop}^{(y)}$ and $H_2=C_y \ket{\eta_y(0)}\bra{\eta_y(0)}$. In this case, we have $g_1=0$, $\Delta_1=O(1/T^2)$ and $g_2=0$, $\Delta_1=C_y\geq 1$. Hence we can take $\Delta=O(1/T^2)$. Moreover, the ground state of $H_{prop}^{(y)}$ is $\ket{\psi^{(y)}}$ whereas the ground state space of $C_y \ket{\eta_y(0)}\bra{\eta_y(0)}$ is spanned by the states $\ket{\eta_y(t)}$, for $t=1,...,T$.

In order to calculate the maximum overlap between the two ground spaces, let us define $\Pi_2=\sum_{t=1}^T \ket{\eta_y(t)}\bra{\eta_y(t)}$. The state belonging to the ground state space of $H_2$ with the maximum overlap with $\ket{\psi^{(y)}}$ is 
$\ket{v_2}=\Pi_2\ket{\psi^{(y)}}/\sqrt{\bra{\psi^{(y)}}\Pi_2\ket{\psi^{(y)}}}$. Hence, we obtain  
\begin{align}
\cos(\theta)&=|\braket{v_2|\psi^{(y)}}|=\sqrt{\bra{\psi^{(y)}}\Pi_2\ket{\psi^{(y)}}}\\&=\sqrt{\frac{T}{T+1}}\leq 1-\frac{1}{2T}.
\end{align}
Hence, the geometrical lemma implies that the lowest energy state of $H^{(y)}$ for $y\notin S^*$ is $O(1/T^3)$.

This shows that the states $\ket{\psi^{(y)}}$ for $y\in S^*$ are the ground states of $H=H_{prop}+H_{pen}$. Consequently, the probability of observing $0$ upon an ideal energy measurement of a quantum state $\ket{x}\ket{T}$ is given by 
\begin{align}
P_{GS}&=\frac{1}{T+1}\sum_{y\in S^*} |\bra{\psi^{(y)}}(\ket{x}\ket{T})|^2\nonumber \\ & = \frac{1}{T+1}\sum_{y\in S^*} |\bra{y}U_T U_{T-1}...U_1\ket{x}|^2. \nonumber\\
&=\frac{p}{T+1}
\end{align}
Let us assume now that we have a classical algorithm for approximate energy sampling for Hamiltonian $H=H_{prop}+H_{pen}$ and initial state $\ket{x}\ket{T}$, with a running time $O(\poly(n,\delta^{-1},\beta^{-1},\epsilon^{-1}))$. In what follows we demonstrate that we can estimate $p$ from such energy sampler by making use of Lemma~\ref{lem:probGS} together with Hoeffding inequality.
First, we choose the parameters of the energy sampler to be $\delta = \Delta/3=O(1/T^3)$, where $\Delta$ is the gap of Hamiltonian $H$ and 
\begin{equation}
\epsilon+\beta=\frac{\delta_p}{2(T+1)},
\end{equation}
where $\delta_p$ is defined in Eq.~\eqref{eq:defpolybox}. By assumption, the energy sampling algorithm would output one sample in time $\poly(n, \delta_p^{-1})$. Given the choice of parameters $\epsilon,\beta$ we obtain from Lemma~\ref{lem:probGS} that the probability of obtaining an outcome $E_m\in [-\Delta/3, \Delta/3]$ is given by $q'_{GS}$ such that 
\begin{equation}\label{eq:errorq}
|q'_{GS}-P_{GS}|\leq \frac{\delta_p}{2(T+1)}.
\end{equation}
We now demonstrate that we can estimate $q'_{GS}$ within an additive error $\delta_p/(2T+2)$ by querying the energy sampler $s$ times and computing the average number of times an event in the interval $[-\Delta/3, \Delta/3]$ is observed. Let us denote this estimator by $\hat{q}_s$. By Hoeffdings inequality we have that
\begin{equation}\label{eq:hoeffdings}
\Pr\left( |q'_{GS}-\hat{q}_s|\geq \frac{\delta_p}{2(T+1)}\right)
\leq 2\exp\left(-\frac{2s{\delta_p}^2}{4(T+1)^2}\right)
\end{equation}
 In order to reach an error of $\epsilon_p$ we choose 
\begin{align}
\exp\left(- \frac{s{\delta_p}^2}{2(T+1)^2}\right)=\epsilon_p\\
\Leftrightarrow s=\log{\left(\frac{2}{\epsilon_p}\right)}\frac{2(T+1)^2}{{\delta_p}^2}\label{eq:nsamples}
\end{align}
Hence, with $s=O(poly(n, \delta_p^{-1}, \log({\epsilon_p}^{-1})))$ number of queries to the energy sampler we can obtain the estimator $\hat{q}_s$ within the desired error bound. 

Finally we can construct our estimator for $p$ as $\hat{p}=(T+1)\hat{q}_s$. Given the choice of $s$ from Eq.~\eqref{eq:nsamples} and using Eqs.~\eqref{eq:errorq} and \eqref{eq:hoeffdings} we obtain 
\begin{equation}
\Pr(|p-\hat{p}|\leq \delta_p)\geq 1-\epsilon_p
\end{equation} 
as desired. The number of samples needed is $s=O(poly(n, \delta_p^{-1}, \log({\epsilon_p}^{-1})))$ and for each sample we require time $\poly(n,\delta_p^{-1})$, which shows that the total running time to compute $\hat{p}$ is $O(poly(n, \delta_p^{-1}, \log({\epsilon_p}^{-1})))$ as required by Definition~\ref{def:polybox}.
\end{proof}

\subsection{Random Energy Measurement (REM) Test}
\label{sec:OpenQuestions}

Given that standard-resolution energy measurements are BQP-hard to simulate, this problem has the potential to be a suitable physically motivated test at which quantum devices can outperform classical simulations. In particular, this suggests the following quantum advantage experiments where one measures the energy of a random local Hamiltonian on an input product state.
\vspace{4pt}
\noindent\hrule
\vspace{5pt}
\noindent\textbf{Random Energy Measurement (REM) Test:}
\begin{enumerate}[leftmargin=13pt]
\item A classical user picks a random many-body local Hamiltonian $H=\sum_i J_i h_i$, where the local terms $\{h_i\}_i$ and couplings $\{J_j\}_j$ are chosen from a target class at random, according to a distribution that is efficient to sample from classically. The latter ensemble is picked so that complexity theoretic evidence against an efficient classical simulation is available.
\item The experimenter performs an approximate standard-resolution measurement of the energy of the Hamiltonian $H$  picked from the ensemble.
\item The test is to produce samples from the output distribution of the above protocol in $O(\poly{(n,\delta^{-1},\beta^{-1}, \epsilon^{-1})})$ time, within a $\beta$ error in the total variation distance.
\end{enumerate}
\hrule
\vspace{5pt}
As discussed in section \ref{sec:Stockmeyer}, this type of test is radically different than usual sampling problems \cite{bosonsampling,IQPapprox}. Further research is thus deemed necessary to fully understand its classical simulability. Below, we discuss open directions for future investigations.

\emph{Complexity of the REM Test.} Theorem \ref{thm:polybox} provides worst-case evidence against the classical simulability of standard-resolution energy measurements. Yet, it provides no insight into the hardness of simulating a \emph{typical} instance of this problem for different ensembles of random local Hamiltonians. Natural candidates that could lead to hard problems on average are, for example, Feynman-Kitaev Hamiltonians encoding random quantum circuits, frustrated spin systems \cite{diep_frustrated_2004} and universal quantum Hamiltonians \cite{cubitt_universal_2018}. However, in order to develop higher confidence against the classical simulability of the REM Test, it would be required to develop new tools to study average-case complexity of problems in BQP. This is because known polynomial-interpolation techniques used in worst-to-average reductions are rather sensitive to noise \cite{bouland_quantum_2018,movassagh_cayley_2019,haferkamp2019closing,movassagh_efficient_2018} and cannot be readily applied in the standard-resolution regime.

\emph{Is the REM test easy to verify?} Commonly-studied quantum sampling problems, with an exponentially large output space, are difficult to verify\mot\cite{Qsup}. Verifying statistical closure in the total variation distance to the ideal distribution based on a single-round of classical post-processing requires exponentially many experimental samples\mot\cite{hangleiter_sample_2019}. Although sample-efficient verification approaches have been proposed\mot\cite{boixo_characterizing_2016,bouland_quantum_2018,Aaronson:2017:CFQ:3135595.3135617}, the verification takes exponential time and works under circuit-level assumptions on noise \cite{boixo_characterizing_2016,Aaronson:2017:CFQ:3135595.3135617,bouland_quantum_2018} or new complexity conjectures \cite{Aaronson:2017:CFQ:3135595.3135617}. If reliable single-qubit measurements are available, a polynomial-time verification is sometimes possible\mot\cite{hangleiter_direct_2017,bermejo,miller_quantum_2017}.

On the other hand, measurements with standard resolution could potentially be easier to verify than commonly-studied sampling problems. Indeed, it is easy to see that they bypass the no-go theorem in Ref.\mot\cite{hangleiter_sample_2019} because of the polynomial size of the output space: 
via the Hoeffding bound, collecting statistics and computing the variation distance to the ideal distribution gives a trivial brute-force exponential-time verification method with polynomial sample complexity, which could be applicable in near-term experiments of limited size. In this context, available  verification methods for BQP-complete problems \cite{verificationElham,verificationfitzsimons2018post,verificationmahadev2018classical,hangleiter_direct_2017} could potentially be useful.
\section{Discussion}\label{sec:discussion}
In this work, we have established energy measurements of many-body Hamiltonians as a problem that can show a reliable quantum advantage based on complexity theoretic arguments. We thus make a key step towards bringing quantum advantage demonstrations closer to physically-motivated questions.

We have analyzed two different regimes regarding the scaling of the cost of performing the measurement, which can be quantified either by the evolution time of the experiment or the number of quantum gates applied in a quantum algorithm such as quantum phase estimation. We have defined a standard-resolution measurement as a measurement where the cost in increasing the resolution is polynomial in $1/\delta$, which is the standard performance of quantum devices for general local Hamiltonians; and super-resolution measurements, where the measurement cost scales as  $\poly(\log(1/\delta))$, which can be achieved by a quantum device if we exploit certain knowledge about the Hamiltonian, such as its diagonalization (or in general, the ability to exponentially fast-forward its time-evolution \cite{aharonov2017}).

We prove that for \emph{super-resolution measurements} it is possible to achieve a quantum advantage demonstration even when the measurement is approximate  (with a system-size-independent sampling error), based on plausible complexity-theoretic assumptions similar to ones used in the "quantum computational supremacy"
literature \cite{bosonsampling,IQPapprox,bermejo,boixo_characterizing_2016,bouland_quantum_2018,haferkamp2019closing}.  
The quantum advantage originates in the super-resolution measurement of a simple 5-local cluster state Hamiltonian on the 2D square lattice on product state inputs. The protocol can be implemented using the  quantum simulation scheme of Ref.~\cite{bermejo} and requires the short time-evolution of a nearest-neighbor on a 2D square lattice, suitable for implementations in, for example, optical lattices. Moreover, this scheme can be efficiently certified using reliable single-qubit measurements. These results open up the possibility of near-term experimental demonstrations of quantum advantage via energy sampling.

In the \emph{standard-resolution regime}, we find two types of complexity-theoretic evidence against the efficient classical simulation of measuring local Hamiltonians. First, in a reminiscent fashion to early work on IQP circuits \cite{IQP}, 
we find a classical simulation to be impossible for simple 2D translation-invariant Hamiltonians in the near-exact sampling regime with inverse-exponential sampling errors, unless the Polynomial Hierarchy collapses. 
Additionally, we point out limitations of available techniques \cite{bosonsampling,IQPapprox} to extend this quantum advantage result to an approximate-sampling one, based on Polynomial Hierarchy collapses. Second, using circuit-to-Hamiltonian constructions and connections to random universal quantum circuits \cite{boixo_characterizing_2016,Aaronson:2017:CFQ:3135595.3135617,bouland_quantum_2018}, we give alternative complexity-theoretic evidence that approximate standard-resolution measurements of 4-local Hamiltonians can show a quantum advantage: a classical simulation here would lead to an efficient classical estimator of marginal probabilities of universal quantum circuits, a BQP-hard task \cite{wocjan2006BQP,janzing2008BQP}.

Three potential improvements related to the technical results in our work are: Firstly, a major challenge would be to tie the hardness of simulating approximate standard-resolution energy measurements to well-known complexity-theoretic conjectures beyond BPP$\neq$BQP. This program would require techniques beyond the Stockmeyer-method and
Polynomial-Hierarchy collapses\cite{bosonsampling,IQPapprox}; Secondly, in this manuscript we have not investigated the verifiability of the standard resolution proposals. However, due to the small size of the energy output space,  classical verification methods similar to those in Refs.\mot\cite{boixo_characterizing_2016,Aaronson:2017:CFQ:3135595.3135617,bouland_quantum_2018} could be developed; Thirdly, it would be interesting to improve the locality of our Hamiltonians. The locality of the examples based on 5-local Hamiltonians on 2D lattices could, in principle, be improved using the general techniques presented in Ref.~\cite{cubitt_universal_2018}, which show that there exist simple 2-local universal Hamiltonians that can reproduce the physics of other Hamiltonians, including the energy spectrum and measurement statistics. The examples based on circuit-to-Hamiltonian constructions could be improved using techniques such as, e.g., perturbation gadgets or space-time circuit-to-Hamiltonian constructions~\cite{oliveira2005, gosset2015universal}.

We have also introduced the concept of 
\textit{quantum Hamiltonian diagonalization} which, up to our knowledge, is a new concept which can be of interest beyond the scope of this work. 
It characterizes a class of Hamiltonians for which there exists a 
polynomial-size quantum circuit $U$ mapping its eigenbasis to the computational basis
and whose eigenvalues can be computed efficiently by a function $f(z)$ on a quantum computer.
This guarantees the exponential fast-forwarding of the dynamics of the Hamiltonian. 


For the purposes of demonstrating quantum advantage, we  restricted ourselves to examples where $f(z)$ can be computed efficiently classically -- this simplifies the protocol for super-resolution measurements so it can be potentially implemented in near-term devices. In this case, the reason why the energy measurement problem is hard to simulate classically results from the fact that the populations of the different eigenstates are \# P-hard to approximate. 
It would be interesting to construct new examples of quantum advantage for super-resolution energy measurements where the classical hardness results from the need to sample from the right eigenvalues (to exponential accuracy) and not only from the right eigenstate populations.

Indeed, Ref.~\cite{aharonov2017} shows that such constructions are, in principle, possible. Therein, the authors present an academic example of a Hamiltonian which can be measured by a quantum algorithm (Shor's algorithm) with super resolution, even though it is not known how to compute its eigenvalues efficiently classically. This Hamiltonian is given by $\hat{H}=U_{ME} + U_{ME}^{\dagger}$, where $U_{ME}$ is the unitary implementation of the modular exponentiation used in Shor's algorithm.  It is interesting to point out that it is possible to find a quantum diagonalization for the aforementioned Hamiltonian using existing quantum algorithms for decomposing finite commutative groups\mot\cite{cheung_decomposing_2001,bermejo-vega_computational_2014}. This academic example shows how quantum algorithms
could potentially be helpful for expanding our knowledge of the inner structure of a given Hamiltonian (here its quantum diagonalization), which can later be exploited to answer specific questions about a given physical system more accurately (here obtaining its spectra). Finding families of Hamiltonians with stronger physical motivation than this example, for which its quantum diagonalization could be learned thanks to a quantum algorithm, would offer a new and interesting application for quantum computers, potentially leading to new exponential speed-ups over known classical algorithms.

Overall, we believe this work brings a new perspective into questions related to Hamiltonian complexity \cite{HamcomplexityCS} by focusing on problems that can be solved efficiently by quantum devices, unlike
problems such as the QMA-complete ground state problem  \cite{QPEkitaev1995}.
Furthermore, we believe it could inspire new demonstrations of quantum advantage for measuring other quantities of interest in quantum many-body physics, which would strengthen the belief that quantum computers and simulators can answer problems about quantum matter beyond the power of any present or future classical algorithms.  
%
%
%
\section*{Acknowledgements}
We acknowledge discussions with Shantanav Chakraborty, Ashley Montanaro, Alex Grilo and Johannes Bausch. Furthermore, we thank Anthony Leverrier, Jelmer Renema and Dominik Hangleiter for valuable feedback on this manuscript. We also thank Mayec Rancel for designing Figure~\ref{Fig:measurement}. LN acknowledges funding from Wiener-Anspach Foundation and F.R.S.-FNRS. J.B.V.\mot acknowledges funding from the European Union’s Horizon 2020 research and innovation programme under the Marie Skłodowska-Curie grant agreement Nº 754446 and UGR Research and Knowledge Transfer Found – Athenea3i. Part of this project took place during J.B.V.'s stay at Freie Universität Berlin, supported by ERC (TAQ). R.G.-P. is a Research Associate of the F.R.S.-FNRS and acknowledges funding from Wiener-Anspach foundation. 
\bibliographystyle{unsrtnat}
\bibliography{Bibliography}

\begin{thebibliography}{102}
\providecommand{\natexlab}[1]{#1}
\providecommand{\url}[1]{\texttt{#1}}
\expandafter\ifx\csname urlstyle\endcsname\relax
  \providecommand{\doi}[1]{doi: #1}\else
  \providecommand{\doi}{doi: \begingroup \urlstyle{rm}\Url}\fi

\bibitem[Gross and Bloch(2017)]{Bloch2017}
Christian Gross and Immanuel Bloch.
\newblock Quantum simulations with ultracold atoms in optical lattices.
\newblock \emph{Science}, 357\penalty0 (6355):\penalty0 995--1001, 2017.
\newblock ISSN 0036-8075.
\newblock \doi{10.1126/science.aal3837}.

\bibitem[Bernien et~al.(2017)Bernien, Schwartz, Keesling, Levine, Omran,
  Pichler, Choi, Zibrov, Endres, Greiner, Vuletić, and Lukin]{Lukin2017}
Hannes Bernien, Sylvain Schwartz, Alexander Keesling, Harry Levine, Ahmed
  Omran, Hannes Pichler, Soonwon Choi, Alexander~S. Zibrov, Manuel Endres,
  Markus Greiner, Vladan Vuletić, and Mikhail~D. Lukin.
\newblock Probing many-body dynamics on a 51-atom quantum simulator.
\newblock \emph{Nature}, 551:\penalty0 579--584, 2017.
\newblock \doi{10.1038/nature24622}.

\bibitem[Zhang et~al.(2017)Zhang, Pagano, Hess, Kyprianidis, Becker, Kaplan,
  Gorshkov, Gong, and Monroe]{Monroe2017}
J.~Zhang, G.~Pagano, P.~W. Hess, A.~Kyprianidis, P.~Becker, H.~Kaplan, A.~V.
  Gorshkov, Z.-X. Gong, and C.~Monroe.
\newblock Observation of a many-body dynamical phase transition with a 53-qubit
  quantum simulator.
\newblock \emph{Nature}, 551:\penalty0 601--604, 2017.
\newblock \doi{10.1038/nature24654}.

\bibitem[Friis et~al.(2018)Friis, Marty, Maier, Hempel, Holzäpfel, Jurcevic,
  Plenio, Huber, Roos, Blatt, and Lanyon]{Lanyon2017}
Nicolai Friis, Oliver Marty, Christine Maier, Cornelius Hempel, Milan
  Holzäpfel, Petar Jurcevic, Martin~B. Plenio, Marcus Huber, Christian Roos,
  Rainer Blatt, and Ben Lanyon.
\newblock Observation of entangled states of a fully controlled 20-qubit
  system.
\newblock \emph{Phys. Rev. X}, 8:\penalty0 021012, 2018.
\newblock \doi{10.1103/PhysRevX.8.021012}.

\bibitem[Neill et~al.(2018)Neill, Roushan, Kechedzhi, Boixo, Isakov,
  Smelyanskiy, Barends, Burkett, Chen, Chen, Chiaro, Dunsworth, Fowler, Foxen,
  Graff, Jeffrey, Kelly, Lucero, Megrant, Mutus, Neeley, Quintana, Sank,
  Vainsencher, Wenner, White, Neven, and Martinis]{Martinis2018}
C.~Neill, P.~Roushan, K.~Kechedzhi, S.~Boixo, S.~V. Isakov, V.~Smelyanskiy,
  R.~Barends, B.~Burkett, Y.~Chen, Z.~Chen, B.~Chiaro, A.~Dunsworth, A.~Fowler,
  B.~Foxen, R.~Graff, E.~Jeffrey, J.~Kelly, E.~Lucero, A.~Megrant, J.~Mutus,
  M.~Neeley, C.~Quintana, D.~Sank, A.~Vainsencher, J.~Wenner, T.~C. White,
  H.~Neven, and J.~M. Martinis.
\newblock A blueprint for demonstrating quantum supremacy with superconducting
  qubits.
\newblock \emph{Science}, pages 195--199, 2018.
\newblock \doi{10.1126/science.aao4309}.

\bibitem[Arute et~al.(2019)Arute, Arya, Babbush, Bacon, Bardin, Barends,
  Biswas, Boixo, Brandao, Buell, et~al.]{googlesupremacy2019}
Frank Arute, Kunal Arya, Ryan Babbush, Dave Bacon, Joseph~C Bardin, Rami
  Barends, Rupak Biswas, Sergio Boixo, Fernando~GSL Brandao, David~A Buell,
  et~al.
\newblock Quantum supremacy using a programmable superconducting processor.
\newblock \emph{Nature}, 574\penalty0 (7779):\penalty0 505--510, 2019.
\newblock \doi{10.1038/s41586-019-1666-5}.

\bibitem[Islam et~al.(2015)Islam, Ma, Preiss, Tai, Lukin, Rispoli, and
  Greiner]{Greiner2015}
Rajibul Islam, Ruichao Ma, Philipp~M. Preiss, M.~Eric Tai, Alexander Lukin,
  Matthew Rispoli, and Markus Greiner.
\newblock Measuring entanglement entropy in a quantum many-body system.
\newblock \emph{Nature}, 528:\penalty0 77–83, 2015.
\newblock \doi{10.1038/nature15750}.

\bibitem[Brydges et~al.(2019)Brydges, Elben, Jurcevic, Vermersch, Maier,
  Lanyon, Zoller, Blatt, and Roos]{Roos2019}
Tiff Brydges, Andreas Elben, Petar Jurcevic, Benoit Vermersch, Christine Maier,
  Ben~P. Lanyon, Peter Zoller, Rainer Blatt, and Christian~F. Roos.
\newblock Probing rényi entanglement entropy via randomized measurements.
\newblock \emph{Science}, 364:\penalty0 260--263, 2019.
\newblock \doi{10.1126/science.aau4963}.

\bibitem[Kaufman et~al.(2016)Kaufman, Tai, Lukin, Rispoli, Schittko, Preiss,
  and Greiner]{Greiner2016}
Adam~M. Kaufman, M.~Eric Tai, Alexander Lukin, Matthew Rispoli, Robert
  Schittko, Philipp~M. Preiss, and Markus Greiner.
\newblock Quantum thermalization through entanglement in an isolated many-body
  system.
\newblock \emph{Science}, 353:\penalty0 794--800, 2016.
\newblock \doi{10.1126/science.aaf6725}.

\bibitem[Keesling et~al.(2019)Keesling, Omran, Levine, Bernien, Pichler, Choi,
  Samajdar, Schwartz, Silvi, Sachdev, Zoller, Endres, Greiner, Vuletić, and
  Lukin]{Lukin2019}
Alexander Keesling, Ahmed Omran, Harry Levine, Hannes Bernien, Hannes Pichler,
  Soonwon Choi, Rhine Samajdar, Sylvain Schwartz, Pietro Silvi, Subir Sachdev,
  Peter Zoller, Manuel Endres, Markus Greiner, Vladan Vuletić, and Mikhail~D.
  Lukin.
\newblock Quantum kibble–zurek mechanism and critical dynamics on a
  programmable rydberg simulator.
\newblock \emph{Nature}, 568:\penalty0 207–211, 2019.
\newblock \doi{10.1038/s41586-019-1070-1}.

\bibitem[Choi et~al.(2016)Choi, Hild, Zeiher, Schau\ss, Rubio-Abadal, Yefsah,
  Khemani, Huse, Bloch, and Gross]{choi_exploring_2016}
Jae-yoon Choi, Sebastian Hild, Johannes Zeiher, Peter Schau\ss, Antonio
  Rubio-Abadal, Tarik Yefsah, Vedika Khemani, David~A. Huse, Immanuel Bloch,
  and Christian Gross.
\newblock Exploring the many-body localization transition in two dimensions.
\newblock \emph{Science}, 352\penalty0 (6293):\penalty0 1547--1552, June 2016.
\newblock \doi{10.1126/science.aaf8834}.

\bibitem[Rispoli et~al.(2019)Rispoli, Lukin, Schittko, Kim, Tai, Léonard, and
  Greiner]{Greiner2019}
Matthew Rispoli, Alexander Lukin, Robert Schittko, Sooshin Kim, M.~Eric Tai,
  Julian Léonard, and Markus Greiner.
\newblock Quantum critical behaviour at the many-body localization transition.
\newblock \emph{Nature}, 2019.
\newblock \doi{10.1038/s41586-019-1527-2}.

\bibitem[Gärttner et~al.(2017)Gärttner, Bohnet, Safavi-Naini, Wall, J., and
  Rey]{Rey2017}
Martin Gärttner, Justin~G. Bohnet, Arghavan Safavi-Naini, Michael~L. Wall,
  Bollinger~John J., and Ana~Maria Rey.
\newblock Measuring out-of-time-order correlations and multiple quantum spectra
  in a trapped-ion quantum magnet.
\newblock \emph{Nature Physics}, 13:\penalty0 781–786, 2017.
\newblock \doi{10.1038/nphys4119}.

\bibitem[Vermersch et~al.(2019)Vermersch, Elben, Sieberer, Yao, and
  Zoller]{Zoller2019}
B.~Vermersch, A.~Elben, L.~M. Sieberer, N.~Y. Yao, and P.~Zoller.
\newblock Probing scrambling using statistical correlations between randomized
  measurements.
\newblock \emph{Phys. Rev. X}, 9:\penalty0 021061, Jun 2019.
\newblock \doi{10.1103/PhysRevX.9.021061}.

\bibitem[Satzinger et~al.(2021)Satzinger, Liu, Smith, Knapp, Newman, Jones,
  Chen, Quintana, Mi, Dunsworth, et~al.]{topological1}
KJ~Satzinger, Y~Liu, A~Smith, C~Knapp, M~Newman, C~Jones, Z~Chen, C~Quintana,
  X~Mi, A~Dunsworth, et~al.
\newblock Realizing topologically ordered states on a quantum processor.
\newblock \emph{arXiv preprint arXiv:2104.01180}, 2021.

\bibitem[Semeghini et~al.(2021)Semeghini, Levine, Keesling, Ebadi, Wang,
  Bluvstein, Verresen, Pichler, Kalinowski, Samajdar, et~al.]{topological2}
Giulia Semeghini, Harry Levine, Alexander Keesling, Sepehr Ebadi, Tout~T Wang,
  Dolev Bluvstein, Ruben Verresen, Hannes Pichler, Marcin Kalinowski, Rhine
  Samajdar, et~al.
\newblock Probing topological spin liquids on a programmable quantum simulator.
\newblock \emph{arXiv preprint arXiv:2104.04119}, 2021.

\bibitem[Trotzky et~al.(2012)Trotzky, Chen, Flesch, McCulloch, Schollw\"ock,
  Eisert, and Bloch]{Trotzky}
S.~Trotzky, Y.-A. Chen, A.~Flesch, I.~P. McCulloch, U.~Schollw\"ock, J.~Eisert,
  and I.~Bloch.
\newblock Probing the relaxation towards equilibrium in an isolated strongly
  correlated one-dimensional {B}ose gas.
\newblock \emph{Nature Phys.}, 8:\penalty0 325--330, 2012.
\newblock \doi{doi:10.1038/nphys2232}.

\bibitem[Braun et~al.(2015)Braun, Friesdorf, Hodgman, Schreiber, Ronzheimer,
  Riera, del Rey, Bloch, Eisert, and Schneider]{Schreiber-pnas-2015}
S.~Braun, M.~Friesdorf, J.~S. Hodgman, M.~Schreiber, J.~P. Ronzheimer,
  A.~Riera, M.~del Rey, I.~Bloch, J.~Eisert, and U.~Schneider.
\newblock Emergence of coherence and the dynamics of quantum phase transitions.
\newblock \emph{PNAS}, 112:\penalty0 3641--3646, March 2015.
\newblock \doi{10.1073/pnas.1408861112}.

\bibitem[Tang(2019)]{tang2019}
Ewin Tang.
\newblock A quantum-inspired classical algorithm for recommendation systems.
\newblock In \emph{Proceedings of the 51st Annual ACM SIGACT Symposium on
  Theory of Computing}, pages 217--228. ACM, 2019.
\newblock \doi{10.1145/3313276.3316310}.

\bibitem[Aaronson and Arkhipov(2011)]{bosonsampling}
Scott Aaronson and Alex Arkhipov.
\newblock The computational complexity of linear optics.
\newblock In \emph{Proceedings of the Forty-Third Annual ACM Symposium on
  Theory of Computing}, STOC '11, page 333–342, New York, NY, USA, 2011.
  Association for Computing Machinery.
\newblock ISBN 9781450306911.
\newblock \doi{10.1145/1993636.1993682}.

\bibitem[Bremner et~al.(2016)Bremner, Montanaro, and Shepherd]{IQPapprox}
Michael~J. Bremner, Ashley Montanaro, and Dan~J. Shepherd.
\newblock Average-case complexity versus approximate simulation of commuting
  quantum computations.
\newblock \emph{Phys. Rev. Lett.}, 117:\penalty0 080501, Aug 2016.
\newblock \doi{10.1103/PhysRevLett.117.080501}.

\bibitem[Boixo et~al.(2018)Boixo, Isakov, Smelyanskiy, Babbush, Ding, Jiang,
  Bremner, Martinis, and Neven]{boixo_characterizing_2016}
Sergio Boixo, Sergei~V. Isakov, Vadim~N. Smelyanskiy, Ryan Babbush, Nan Ding,
  Zhang Jiang, Michael~J. Bremner, John~M. Martinis, and Hartmut Neven.
\newblock Characterizing quantum supremacy in near-term devices.
\newblock \emph{Nature Physics}, page~1, April 2018.
\newblock ISSN 1745-2481.
\newblock \doi{10.1038/s41567-018-0124-x}.

\bibitem[Gao et~al.(2017)Gao, Wang, and Duan]{gao}
Xun Gao, Sheng-Tao Wang, and L.-M. Duan.
\newblock Quantum supremacy for simulating a translation-invariant ising spin
  model.
\newblock \emph{Phys. Rev. Lett.}, 118:\penalty0 040502, Jan 2017.
\newblock \doi{10.1103/PhysRevLett.118.040502}.

\bibitem[Bermejo-Vega et~al.(2018)Bermejo-Vega, Hangleiter, Schwarz,
  Raussendorf, and Eisert]{bermejo}
Juan Bermejo-Vega, Dominik Hangleiter, Martin Schwarz, Robert Raussendorf, and
  Jens Eisert.
\newblock Architectures for quantum simulation showing a quantum speedup.
\newblock \emph{Phys. Rev. X}, 8:\penalty0 021010, Apr 2018.
\newblock \doi{10.1103/PhysRevX.8.021010}.

\bibitem[Harrow and Montanaro(2017)]{Qsup}
Aram~W Harrow and Ashley Montanaro.
\newblock Quantum computational supremacy.
\newblock \emph{Nature}, 549\penalty0 (7671):\penalty0 203, 2017.
\newblock \doi{10.1038/nature23458}.

\bibitem[Preskill(2018)]{preskill}
John Preskill.
\newblock Quantum {C}omputing in the {NISQ} era and beyond.
\newblock \emph{{Quantum}}, 2:\penalty0 79, August 2018.
\newblock ISSN 2521-327X.
\newblock \doi{10.22331/q-2018-08-06-79}.

\bibitem[Bremner et~al.(2017)Bremner, Montanaro, and
  Shepherd]{bremner_achieving_2017}
Michael~J. Bremner, Ashley Montanaro, and Dan~J. Shepherd.
\newblock Achieving quantum supremacy with sparse and noisy commuting quantum
  computations.
\newblock \emph{Quantum}, 1:\penalty0 8, April 2017.
\newblock \doi{10.22331/q-2017-04-25-8}.

\bibitem[Miller et~al.(2017)Miller, Sanders, and Miyake]{miller_quantum_2017}
J.~Miller, S.~Sanders, and A.~Miyake.
\newblock Quantum supremacy in constant-time measurement-based computation: {A}
  unified architecture for sampling and verification.
\newblock \emph{Phys. Rev. A}, 96\penalty0 (6):\penalty0 062320, December 2017.
\newblock \doi{10.1103/PhysRevA.96.062320}.

\bibitem[Aaronson and Chen(2017)]{Aaronson:2017:CFQ:3135595.3135617}
Scott Aaronson and Lijie Chen.
\newblock Complexity-theoretic foundations of quantum supremacy experiments.
\newblock In \emph{Proceedings of the 32Nd Computational Complexity
  Conference}, CCC '17, pages 22:1--22:67, Germany, 2017. Schloss
  Dagstuhl--Leibniz-Zentrum fuer Informatik.
\newblock ISBN 978-3-95977-040-8.
\newblock \doi{10.4230/LIPIcs.CCC.2017.22}.

\bibitem[Hangleiter et~al.(2017)Hangleiter, Kliesch, Schwarz, and
  Eisert]{hangleiter_direct_2017}
D.~Hangleiter, M.~Kliesch, M.~Schwarz, and J.~Eisert.
\newblock Direct certification of a class of quantum simulations.
\newblock \emph{Quantum Sci. Technol.}, 2\penalty0 (1):\penalty0 015004, 2017.
\newblock ISSN 2058-9565.
\newblock \doi{10.1088/2058-9565/2/1/015004}.

\bibitem[Bouland et~al.(2019)Bouland, Fefferman, Nirkhe, and
  Vazirani]{bouland_quantum_2018}
A.~Bouland, B.~Fefferman, C.~Nirkhe, and U.~Vazirani.
\newblock Quantum {supremacy} and the {complexity} of {random} {circuit}
  {sampling}.
\newblock \emph{Nature Phys.}, 15:\penalty0 159--163, March 2019.
\newblock \doi{10.1038/s41567-018-0318-2}.

\bibitem[Deshpande et~al.(2018)Deshpande, Fefferman, Tran, Foss-Feig, and
  Gorshkov]{deshpande_dynamical_2018}
Abhinav Deshpande, Bill Fefferman, Minh~C. Tran, Michael Foss-Feig, and
  Alexey~V. Gorshkov.
\newblock Dynamical {Phase} {Transitions} in {Sampling} {Complexity}.
\newblock \emph{Physical Review Letters}, 121\penalty0 (3):\penalty0 030501,
  July 2018.
\newblock \doi{10.1103/PhysRevLett.121.030501}.

\bibitem[Muraleedharan et~al.(2019)Muraleedharan, Miyake, and
  Deutsch]{muraleedharan_quantum_2019}
Gopikrishnan Muraleedharan, Akimasa Miyake, and Ivan~H. Deutsch.
\newblock Quantum computational supremacy in the sampling of bosonic random
  walkers on a one-dimensional lattice.
\newblock \emph{New Journal of Physics}, 21\penalty0 (5):\penalty0 055003, May
  2019.
\newblock ISSN 1367-2630.
\newblock \doi{10.1088/1367-2630/ab0610}.
\newblock arXiv: 1805.01858.

\bibitem[Maskara et~al.(2019)Maskara, Deshpande, Tran, Ehrenberg, Fefferman,
  and Gorshkov]{maskara_complexity_2019}
Nishad Maskara, Abhinav Deshpande, Minh~C. Tran, Adam Ehrenberg, Bill
  Fefferman, and Alexey~V. Gorshkov.
\newblock Complexity phase diagram for interacting and long-range bosonic
  {Hamiltonians}.
\newblock \emph{arXiv:1906.04178}, June 2019.

\bibitem[Kitaev(1995)]{QPEkitaev1995}
A~Yu Kitaev.
\newblock Quantum measurements and the abelian stabilizer problem.
\newblock \emph{arXiv:quant-ph/9511026}, 1995.

\bibitem[Abrams and Lloyd(1999)]{abrams1999quantum}
Daniel~S. Abrams and Seth Lloyd.
\newblock Quantum algorithm providing exponential speed increase for finding
  eigenvalues and eigenvectors.
\newblock \emph{Phys. Rev. Lett.}, 83:\penalty0 5162--5165, Dec 1999.
\newblock \doi{10.1103/PhysRevLett.83.5162}.

\bibitem[Yang et~al.(2020)Yang, Grankin, Sieberer, Vasilyev, and
  Zoller]{QNDzoller}
Dayou Yang, Andrey Grankin, Lukas~M Sieberer, Denis~V Vasilyev, and Peter
  Zoller.
\newblock Quantum non-demolition measurement of a many-body hamiltonian.
\newblock \emph{Nature communications}, 11\penalty0 (1):\penalty0 1--8, 2020.
\newblock \doi{10.1038/s41467-020-14489-5}.

\bibitem[Huh et~al.(2015)Huh, Guerreschi, Peropadre, McClean, and
  Aspuru-Guzik]{vibronicspectra}
Joonsuk Huh, Gian~Giacomo Guerreschi, Borja Peropadre, Jarrod~R McClean, and
  Al{\'a}n Aspuru-Guzik.
\newblock Boson sampling for molecular vibronic spectra.
\newblock \emph{Nature Photonics}, 9\penalty0 (9):\penalty0 615, 2015.
\newblock \doi{10.1038/nphoton.2015.153}.

\bibitem[Hamilton et~al.(2017)Hamilton, Kruse, Sansoni, Barkhofen, Silberhorn,
  and Jex]{Hamilton2017}
Craig~S. Hamilton, Regina Kruse, Linda Sansoni, Sonja Barkhofen, Christine
  Silberhorn, and Igor Jex.
\newblock Gaussian boson sampling.
\newblock \emph{Phys. Rev. Lett.}, 119:\penalty0 170501, Oct 2017.
\newblock \doi{10.1103/PhysRevLett.119.170501}.

\bibitem[(Editor)(2012)]{Barone12}
Vincenzo~Barone (Editor).
\newblock \emph{Computational Strategies for Spectroscopy: from Small Molecules
  to Nano Systems}.
\newblock Wiley, Hoboken, New Jersey, USA, 2012.
\newblock ISBN 0470470178.

\bibitem[Santoro et~al.(2007)Santoro, Improta, Lami, Bloino, and
  Barone]{Santoro2007}
Fabrizio Santoro, Roberto Improta, Alessandro Lami, Julien Bloino, and Vincenzo
  Barone.
\newblock Effective method to compute franck-condon integrals for optical
  spectra of large molecules in solution.
\newblock \emph{The Journal of Chemical Physics}, 126\penalty0 (8):\penalty0
  084509, 2007.
\newblock \doi{10.1063/1.2437197}.

\bibitem[Wocjan and Zhang(2006)]{wocjan2006BQP}
Pawel Wocjan and Shengyu Zhang.
\newblock Several natural bqp-complete problems.
\newblock \emph{arXiv:quant-ph/0606179}, 2006.

\bibitem[Janzing et~al.(2008)Janzing, Wocjan, and Zhang]{janzing2008BQP}
Dominik Janzing, Pawel Wocjan, and Shengyu Zhang.
\newblock A single-shot measurement of the energy of product states in a
  translation invariant spin chain can replace any quantum computation.
\newblock \emph{New Journal of Physics}, 10\penalty0 (9):\penalty0 093004,
  2008.
\newblock \doi{10.1088/1367-2630/10/9/093004}.

\bibitem[Atia and Aharonov(2017)]{aharonov2017}
Yosi Atia and Dorit Aharonov.
\newblock Fast-forwarding of hamiltonians and exponentially precise
  measurements.
\newblock \emph{Nature communications}, 8\penalty0 (1):\penalty0 1572, 2017.
\newblock \doi{10.1038/s41467-017-01637-7}.

\bibitem[Von~Neumann(2018)]{vonneumann}
John Von~Neumann.
\newblock \emph{Mathematical foundations of quantum mechanics: New edition}.
\newblock Princeton university press, 2018.

\bibitem[Aharonov et~al.(2002)Aharonov, Massar, and Popescu]{AMPmeasuring2002}
Y.~Aharonov, S.~Massar, and S.~Popescu.
\newblock Measuring energy, estimating hamiltonians, and the time-energy
  uncertainty relation.
\newblock \emph{Phys. Rev. A}, 66:\penalty0 052107, Nov 2002.
\newblock \doi{10.1103/PhysRevA.66.052107}.

\bibitem[Wocjan et~al.(2003)Wocjan, Janzing, Decker, and
  Beth]{wocjan2003PSPACE}
Pawel Wocjan, Dominik Janzing, Thomas Decker, and Thomas Beth.
\newblock Measuring 4-local n-qubit observables could probabilistically solve
  {PSPACE}.
\newblock \emph{arXiv:quant-ph/0308011}, 2003.

\bibitem[Shepherd and Bremner(2009)]{IQP}
Dan Shepherd and Michael~J. Bremner.
\newblock Temporally unstructured quantum computation.
\newblock \emph{Proceedings of the Royal Society of London A: Mathematical,
  Physical and Engineering Sciences}, 465\penalty0 (2105):\penalty0 1413--1439,
  2009.
\newblock ISSN 1364-5021.
\newblock \doi{10.1098/rspa.2008.0443}.

\bibitem[Morimae et~al.(2014)Morimae, Fujii, and
  Fitzsimons]{Morimae14HardnessDQC1}
Tomoyuki Morimae, Keisuke Fujii, and Joseph~F. Fitzsimons.
\newblock Hardness of classically simulating the one-clean-qubit model.
\newblock \emph{Phys. Rev. Lett.}, 112:\penalty0 130502, Apr 2014.
\newblock \doi{10.1103/PhysRevLett.112.130502}.

\bibitem[Morimae(2017)]{Morimae1704.03640}
Tomoyuki Morimae.
\newblock Hardness of classically sampling the one-clean-qubit model with
  constant total variation distance error.
\newblock \emph{Phys. Rev. A}, 96:\penalty0 040302, Oct 2017.
\newblock \doi{10.1103/PhysRevA.96.040302}.

\bibitem[Kitaev et~al.(2002)Kitaev, Shen, and Vyalyi]{kitaev2002}
Alexei~Yu Kitaev, Alexander Shen, and Mikhail~N Vyalyi.
\newblock \emph{Classical and quantum computation}.
\newblock Number~47. American Mathematical Soc., 2002.

\bibitem[Childs et~al.(2002)Childs, Deotto, Farhi, Goldstone, Gutmann, and
  Landahl]{childsmeasurement}
Andrew~M. Childs, Enrico Deotto, Edward Farhi, Jeffrey Goldstone, Sam Gutmann,
  and Andrew~J. Landahl.
\newblock Quantum search by measurement.
\newblock \emph{Phys. Rev. A}, 66:\penalty0 032314, Sep 2002.
\newblock \doi{10.1103/PhysRevA.66.032314}.

\bibitem[Lloyd(1996)]{lloyd1996universal}
Seth Lloyd.
\newblock Universal quantum simulators.
\newblock \emph{Science}, pages 1073--1078, 1996.
\newblock \doi{10.1126/science.273.5278.1073}.

\bibitem[Berry et~al.(2015{\natexlab{a}})Berry, Childs, Cleve, Kothari, and
  Somma]{truncatedtaylor}
Dominic~W. Berry, Andrew~M. Childs, Richard Cleve, Robin Kothari, and
  Rolando~D. Somma.
\newblock Simulating hamiltonian dynamics with a truncated taylor series.
\newblock \emph{Phys. Rev. Lett.}, 114:\penalty0 090502, Mar
  2015{\natexlab{a}}.
\newblock \doi{10.1103/PhysRevLett.114.090502}.

\bibitem[Somma et~al.(2002)Somma, Ortiz, Gubernatis, Knill, and
  Laflamme]{somma2002}
R.~Somma, G.~Ortiz, J.~E. Gubernatis, E.~Knill, and R.~Laflamme.
\newblock Simulating physical phenomena by quantum networks.
\newblock \emph{Phys. Rev. A}, 65:\penalty0 042323, Apr 2002.
\newblock \doi{10.1103/PhysRevA.65.042323}.

\bibitem[Somma(2019)]{somma2019}
Rolando~D Somma.
\newblock Quantum eigenvalue estimation via time series analysis.
\newblock \emph{New Journal of Physics}, 21\penalty0 (12):\penalty0 123025,
  2019.
\newblock \doi{10.1088/1367-2630/ab5c60}.

\bibitem[Roushan et~al.(2017)Roushan, Neill, Tangpanitanon, Bastidas, Megrant,
  Barends, Chen, Chen, Chiaro, Dunsworth, Fowler, Foxen, Giustina, Jeffrey,
  Kelly, Lucero, Mutus, Neeley, Quintana, Sank, Vainsencher, Wenner, White,
  Neven, Angelakis, and Martinis]{Roushanmanybody}
P.~Roushan, C.~Neill, J.~Tangpanitanon, V.~M. Bastidas, A.~Megrant, R.~Barends,
  Y.~Chen, Z.~Chen, B.~Chiaro, A.~Dunsworth, A.~Fowler, B.~Foxen, M.~Giustina,
  E.~Jeffrey, J.~Kelly, E.~Lucero, J.~Mutus, M.~Neeley, C.~Quintana, D.~Sank,
  A.~Vainsencher, J.~Wenner, T.~White, H.~Neven, D.~G. Angelakis, and
  J.~Martinis.
\newblock Spectroscopic signatures of localization with interacting photons in
  superconducting qubits.
\newblock \emph{Science}, 358\penalty0 (6367):\penalty0 1175--1179, 2017.
\newblock ISSN 0036-8075.
\newblock \doi{10.1126/science.aao1401}.

\bibitem[Terhal and DiVincenzo(2004)]{constantdepth}
B.~M. Terhal and D.~P. DiVincenzo.
\newblock Adaptive quantum computation, constant depth quantum circuits and
  arthur-merlin games.
\newblock \emph{Quantum Inf. Comput.}, 4\penalty0 (2):\penalty0 134--145, 2004.

\bibitem[Bremner et~al.(2010)Bremner, Jozsa, and Shepherd]{IQPexact}
Michael~J. Bremner, Richard Jozsa, and Dan~J. Shepherd.
\newblock Classical simulation of commuting quantum computations implies
  collapse of the polynomial hierarchy.
\newblock \emph{Proceedings of the Royal Society of London A: Mathematical,
  Physical and Engineering Sciences}, 2010.
\newblock ISSN 1364-5021.
\newblock \doi{10.1098/rspa.2010.0301}.

\bibitem[Berry et~al.(2015{\natexlab{b}})Berry, Childs, and Kothari]{BCK2015}
D.~W. Berry, A.~M. Childs, and R.~Kothari.
\newblock Hamiltonian simulation with nearly optimal dependence on all
  parameters.
\newblock In \emph{2015 IEEE 56th Annual Symposium on Foundations of Computer
  Science}, pages 792--809, Oct 2015{\natexlab{b}}.
\newblock \doi{10.1109/FOCS.2015.54}.

\bibitem[Low and Chuang(2017)]{low2017}
Guang~Hao Low and Isaac~L. Chuang.
\newblock Optimal hamiltonian simulation by quantum signal processing.
\newblock \emph{Phys. Rev. Lett.}, 118:\penalty0 010501, Jan 2017.
\newblock \doi{10.1103/PhysRevLett.118.010501}.

\bibitem[Low and Chuang(2019)]{low2016}
Guang~Hao Low and Isaac~L. Chuang.
\newblock Hamiltonian {S}imulation by {Q}ubitization.
\newblock \emph{{Quantum}}, 3:\penalty0 163, July 2019.
\newblock ISSN 2521-327X.
\newblock \doi{10.22331/q-2019-07-12-163}.

\bibitem[Chakraborty et~al.(2019)Chakraborty, Gily{\'e}n, and
  Jeffery]{chakraborty2018}
Shantanav Chakraborty, Andr{\'a}s Gily{\'e}n, and Stacey Jeffery.
\newblock {The Power of Block-Encoded Matrix Powers: Improved Regression
  Techniques via Faster Hamiltonian Simulation}.
\newblock In Christel Baier, Ioannis Chatzigiannakis, Paola Flocchini, and
  Stefano Leonardi, editors, \emph{46th International Colloquium on Automata,
  Languages, and Programming (ICALP 2019)}, volume 132 of \emph{Leibniz
  International Proceedings in Informatics (LIPIcs)}, pages 33:1--33:14,
  Dagstuhl, Germany, 2019. Schloss Dagstuhl--Leibniz-Zentrum fuer Informatik.
\newblock ISBN 978-3-95977-109-2.
\newblock \doi{10.4230/LIPIcs.ICALP.2019.33}.

\bibitem[Aharonov and Ben-Or(1997)]{aharonov1999fault}
D.~Aharonov and M.~Ben-Or.
\newblock Fault-tolerant quantum computation with constant error.
\newblock In \emph{Proceedings of the Twenty-Ninth Annual ACM Symposium on
  Theory of Computing}, STOC '97, page 176–188, New York, NY, USA, 1997.
  Association for Computing Machinery.
\newblock ISBN 0897918886.
\newblock \doi{10.1145/258533.258579}.

\bibitem[Nielsen and Chuang(2002)]{nielsen2002quantum}
Michael~A Nielsen and Isaac Chuang.
\newblock Quantum computation and quantum information, 2002.

\bibitem[Blaizot and Ripka(1986)]{blaizot1986quantum}
Jean-Paul Blaizot and Georges Ripka.
\newblock \emph{Quantum theory of finite systems}, volume~3.
\newblock MIT press Cambridge, MA, 1986.

\bibitem[Shchesnovich(2013)]{shchesnovich2013}
VS~Shchesnovich.
\newblock The second quantization method for indistinguishable particles
  (lecture notes in physics, ufabc 2010).
\newblock \emph{arXiv preprint arXiv:1308.3275}, 2013.

\bibitem[Cirstoiu et~al.(2020)Cirstoiu, Holmes, Iosue, Cincio, Coles, and
  Sornborger]{variationalFF}
Cristina Cirstoiu, Zoe Holmes, Joseph Iosue, Lukasz Cincio, Patrick~J Coles,
  and Andrew Sornborger.
\newblock Variational fast forwarding for quantum simulation beyond the
  coherence time.
\newblock \emph{npj Quantum Information}, 6\penalty0 (1):\penalty0 1--10, 2020.
\newblock \doi{10.1038/s41534-020-00302-0}.

\bibitem[Gottesman(1997)]{GottesmanThesis}
Daniel Gottesman.
\newblock Stabilizer codes and quantum error correction.
\newblock \emph{arXiv preprint quant-ph/9705052}, 1997.
\newblock \doi{10.7907/rzr7-dt72}.

\bibitem[Briegel and Raussendorf(2001)]{Briegel-PRL-2001}
H.~J. Briegel and R.~Raussendorf.
\newblock Persistent entanglement in arrays of interacting particles.
\newblock \emph{Phys. Rev. Lett.}, 86:\penalty0 910, 2001.

\bibitem[Raussendorf and Briegel(2001)]{RaussendorfPhysRevLett.86.5188}
R.~Raussendorf and H.~J. Briegel.
\newblock A one-way quantum computer.
\newblock \emph{Phys. Rev. Lett.}, 86:\penalty0 5188--5191, May 2001.
\newblock \doi{10.1103/PhysRevLett.86.5188}.

\bibitem[Hangleiter et~al.(2018)Hangleiter, Bermejo-Vega, Schwarz, and
  Eisert]{Hangleiter2018anticoncentration}
Dominik Hangleiter, Juan Bermejo-Vega, Martin Schwarz, and Jens Eisert.
\newblock Anticoncentration theorems for schemes showing a quantum speedup.
\newblock \emph{{Quantum}}, 2:\penalty0 65, May 2018.
\newblock ISSN 2521-327X.
\newblock \doi{10.22331/q-2018-05-22-65}.

\bibitem[Karp and Lipton(1980)]{karp1980}
R.~M. Karp and R.~J. Lipton.
\newblock Some connections between nonuniform and uniform complexity classes.
\newblock In \emph{Proc. 12th Annu. Symp. Theory Comput.}, STOC, pages
  302--309. ACM, 1980.
\newblock ISBN 0-89791-017-6.
\newblock \doi{10.1145/800141.804678}.

\bibitem[Fortnow(2005)]{fortnow2005beyond}
L.~Fortnow.
\newblock Beyond {NP}: The work and legacy of {L}arry {S}tockmeyer.
\newblock In \emph{Proc. 37th Annu. Symp. Theory Comput.}, STOC, pages
  120--127. ACM, 2005.
\newblock \doi{10.1145/1060590.1060609}.

\bibitem[Aaronson(2016)]{aaronson2016}
S.~Aaronson.
\newblock \emph{{P$\neq$NP?}}
\newblock Springer, 2016.
\newblock \doi{10.1007/978-3-319-32162-2}.

\bibitem[Stockmeyer(1985)]{stockmeyer1985approximation}
Larry Stockmeyer.
\newblock On approximation algorithms for \# {P}.
\newblock \emph{SIAM Journal on Computing}, 14\penalty0 (4):\penalty0 849--861,
  1985.
\newblock \doi{10.1137/0214060}.

\bibitem[Toda(1991)]{Toda:1991:PHP:125944.125952}
Seinosuke Toda.
\newblock {PP} is as hard as the polynomial-time hierarchy.
\newblock \emph{SIAM J. Comput.}, 20\penalty0 (5):\penalty0 865--877, October
  1991.
\newblock ISSN 0097-5397.
\newblock \doi{10.1137/0220053}.

\bibitem[Haferkamp et~al.(2020)Haferkamp, Hangleiter, Bouland, Fefferman,
  Eisert, and Bermejo-Vega]{haferkamp2019closing}
J.~Haferkamp, D.~Hangleiter, A.~Bouland, B.~Fefferman, J.~Eisert, and
  J.~Bermejo-Vega.
\newblock Closing gaps of a quantum advantage with short-time hamiltonian
  dynamics.
\newblock \emph{Phys. Rev. Lett.}, 125:\penalty0 250501, Dec 2020.
\newblock \doi{10.1103/PhysRevLett.125.250501}.

\bibitem[Aaronson(2005)]{Aaronson-ProcRS-2005}
S.~Aaronson.
\newblock Quantum computing, postselection, and probabilistic polynomial-time.
\newblock \emph{Proc. Roy. Soc. A}, 461:\penalty0 2063, 2005.
\newblock ISSN 1364-5021.
\newblock \doi{10.1098/rspa.2005.1546}.

\bibitem[Aharonov et~al.(2008)Aharonov, Van~Dam, Kempe, Landau, Lloyd, and
  Regev]{aharonov2008adiabatic}
Dorit Aharonov, Wim Van~Dam, Julia Kempe, Zeph Landau, Seth Lloyd, and Oded
  Regev.
\newblock Adiabatic quantum computation is equivalent to standard quantum
  computation.
\newblock \emph{SIAM review}, 50\penalty0 (4):\penalty0 755--787, 2008.
\newblock \doi{10.1137/080734479}.

\bibitem[Sachdev(2011)]{sachdev11}
Subir Sachdev.
\newblock \emph{Quantum Phase Transitions}.
\newblock Cambridge University Press, 2 edition, 2011.
\newblock \doi{10.1017/CBO9780511973765}.

\bibitem[Huang and Chen(2015)]{YichenChen15}
Yichen Huang and Xie Chen.
\newblock Quantum circuit complexity of one-dimensional topological phases.
\newblock \emph{Phys. Rev. B}, 91:\penalty0 195143, May 2015.
\newblock \doi{10.1103/PhysRevB.91.195143}.

\bibitem[Caha et~al.(2018)Caha, Landau, and Nagaj]{caha2018}
Libor Caha, Zeph Landau, and Daniel Nagaj.
\newblock Clocks in feynman's computer and kitaev's local hamiltonian: Bias,
  gaps, idling, and pulse tuning.
\newblock \emph{Phys. Rev. A}, 97:\penalty0 062306, Jun 2018.
\newblock \doi{10.1103/PhysRevA.97.062306}.

\bibitem[Ni and Van Den~Nest(2013)]{ni2012commuting}
Xiaotong Ni and Maarten Van Den~Nest.
\newblock Commuting quantum circuits: efficient classical simulations versus
  hardness results.
\newblock \emph{Quantum Information \& Computation}, 13\penalty0
  (1-2):\penalty0 54--72, 2013.

\bibitem[Moylett and Turner(2018)]{Qcircuitbosonsampling}
Alexandra~E. Moylett and Peter~S. Turner.
\newblock Quantum simulation of partially distinguishable boson sampling.
\newblock \emph{Phys. Rev. A}, 97:\penalty0 062329, Jun 2018.
\newblock \doi{10.1103/PhysRevA.97.062329}.

\bibitem[Bennett et~al.(1997)Bennett, Bernstein, Brassard, and
  Vazirani]{Bennet97StrenghtsWeaknnesses_QC}
C.~H. Bennett, E.~Bernstein, G.~Brassard, and U.~Vazirani.
\newblock Strengths and weaknesses of quantum computing.
\newblock \emph{SIAM J. Comp.}, 26\penalty0 (5):\penalty0 1510--1523, 1997.
\newblock \doi{10.1137/S0097539796300933}.

\bibitem[Pashayan et~al.(2020)Pashayan, Bartlett, and Gross]{pashayan2017}
Hakop Pashayan, Stephen~D. Bartlett, and David Gross.
\newblock From estimation of quantum probabilities to simulation of quantum
  circuits.
\newblock \emph{{Quantum}}, 4:\penalty0 223, January 2020.
\newblock ISSN 2521-327X.
\newblock \doi{10.22331/q-2020-01-13-223}.

\bibitem[Aharonov and Naveh(2002)]{aharonovsurvey}
Dorit Aharonov and Tomer Naveh.
\newblock Quantum {NP}-a survey.
\newblock \emph{arXiv preprint quant-ph/0210077}, 2002.

\bibitem[Diep(2004)]{diep_frustrated_2004}
H.~T. Diep.
\newblock \emph{Frustrated {Spin} {Systems}}.
\newblock World Scientific, 2004.
\newblock ISBN 978-981-256-781-9.

\bibitem[Cubitt et~al.(2018)Cubitt, Montanaro, and
  Piddock]{cubitt_universal_2018}
Toby~S. Cubitt, Ashley Montanaro, and Stephen Piddock.
\newblock Universal quantum {Hamiltonians}.
\newblock \emph{Proceedings of the National Academy of Sciences}, 115\penalty0
  (38):\penalty0 9497--9502, September 2018.
\newblock ISSN 0027-8424, 1091-6490.
\newblock \doi{10.1073/pnas.1804949115}.

\bibitem[Movassagh(2019)]{movassagh_cayley_2019}
Ramis Movassagh.
\newblock Cayley path and quantum computational supremacy: {A} proof of
  average-case \#{P}-hardness of {Random} {Circuit} {Sampling} with quantified
  robustness.
\newblock \emph{arXiv:1909.06210}, September 2019.

\bibitem[Movassagh(2018)]{movassagh_efficient_2018}
Ramis Movassagh.
\newblock Efficient unitary paths and quantum computational supremacy: {A}
  proof of average-case hardness of {Random} {Circuit} {Sampling}.
\newblock \emph{arXiv:1810.04681}, October 2018.

\bibitem[Hangleiter et~al.(2019)Hangleiter, Kliesch, Eisert, and
  Gogolin]{hangleiter_sample_2019}
Dominik Hangleiter, Martin Kliesch, Jens Eisert, and Christian Gogolin.
\newblock Sample complexity of device-independently certified ``quantum
  supremacy''.
\newblock \emph{Phys. Rev. Lett.}, 122:\penalty0 210502, May 2019.
\newblock \doi{10.1103/PhysRevLett.122.210502}.

\bibitem[Fitzsimons and Kashefi(2017)]{verificationElham}
Joseph~F. Fitzsimons and Elham Kashefi.
\newblock Unconditionally verifiable blind quantum computation.
\newblock \emph{Phys. Rev. A}, 96:\penalty0 012303, Jul 2017.
\newblock \doi{10.1103/PhysRevA.96.012303}.

\bibitem[Fitzsimons et~al.(2018)Fitzsimons, Hajdu\ifmmode~\check{s}\else
  \v{s}\fi{}ek, and Morimae]{verificationfitzsimons2018post}
Joseph~F. Fitzsimons, Michal Hajdu\ifmmode~\check{s}\else \v{s}\fi{}ek, and
  Tomoyuki Morimae.
\newblock Post hoc verification of quantum computation.
\newblock \emph{Phys. Rev. Lett.}, 120:\penalty0 040501, Jan 2018.
\newblock \doi{10.1103/PhysRevLett.120.040501}.

\bibitem[Mahadev(2018)]{verificationmahadev2018classical}
Urmila Mahadev.
\newblock Classical verification of quantum computations.
\newblock In \emph{2018 IEEE 59th Annual Symposium on Foundations of Computer
  Science (FOCS)}, pages 259--267. IEEE, 2018.
\newblock \doi{10.1109/FOCS.2018.00033}.

\bibitem[Oliveira and Terhal(2008)]{oliveira2005}
Roberto Oliveira and Barbara~M Terhal.
\newblock The complexity of quantum spin systems on a two-dimensional square
  lattice.
\newblock \emph{Quantum Information \& Computation}, 8\penalty0 (10):\penalty0
  900--924, 2008.

\bibitem[Gosset et~al.(2015)Gosset, Terhal, and
  Vershynina]{gosset2015universal}
David Gosset, Barbara~M. Terhal, and Anna Vershynina.
\newblock Universal adiabatic quantum computation via the space-time
  circuit-to-hamiltonian construction.
\newblock \emph{Phys. Rev. Lett.}, 114:\penalty0 140501, Apr 2015.
\newblock \doi{10.1103/PhysRevLett.114.140501}.

\bibitem[Cheung and Mosca(2001)]{cheung_decomposing_2001}
Kevin K.~H. Cheung and Michele Mosca.
\newblock Decomposing {Finite} {Abelian} {Groups}.
\newblock \emph{Quantum Info. Comput.}, 1\penalty0 (3):\penalty0 26--32,
  October 2001.
\newblock ISSN 1533-7146.
\newblock \doi{10.5555/2011339.2011341}.

\bibitem[Bermejo-Vega et~al.()Bermejo-Vega, Lin, and
  Nest]{bermejo-vega_computational_2014}
Juan Bermejo-Vega, Cedric Yen-Yu Lin, and Maarten Van~den Nest.
\newblock The computational power of normalizer circuits over black-box groups.
\newblock \emph{arXiv:1409.4800}.

\bibitem[Gharibian et~al.(2015)Gharibian, Huang, Landau, and
  Shin]{HamcomplexityCS}
Sevag Gharibian, Yichen Huang, Zeph Landau, and Seung~Woo Shin.
\newblock Quantum hamiltonian complexity.
\newblock \emph{Foundations and Trends® in Theoretical Computer Science},
  10\penalty0 (3):\penalty0 159--282, 2015.
\newblock ISSN 1551-305X.
\newblock \doi{10.1561/0400000066}.

\bibitem[Kuperberg(2015)]{Kuperberg15}
Greg Kuperberg.
\newblock How hard is it to approximate the jones polynomial?
\newblock \emph{Theory of Computing}, 11\penalty0 (6):\penalty0 183--219, 2015.
\newblock \doi{10.4086/toc.2015.v011a006}.

\end{thebibliography}

\appendix

\section{Relation between fast-forwarding of Hamiltonians and exponentially precise energy measurements}\label{app:fastforwarding}

The work of Atai-Aharonov \cite{aharonov2017} demonstrates a fundamental relation between the ability to fast-forward a Hamiltonian and the ability to do exponentially precise energy measurements. In this section, we summarize the definitions of Ref.~\cite{aharonov2017} and explain how our results fit in the context of that work. 

A normalized Hamiltonian $H$ is said to be \emph{exponentially fast-forwardable} if a poly-size quantum circuit $U'$ can be constructed such that $||U'-\exp(-iHT)||\leq \alpha$ for $T=O(2^{\Omega(n)})$ and $\alpha=1/\poly(n)$.
Atai-Aharonov show that the ability to exponentially fast-forward a Hamiltonian implies that one can find a poly-size circuit $\tilde{U}_{EM}$ such that $||\tilde{U}_{EM}-U_{EM}|| \leq \alpha'$, where $U_{EM}$ is a unitary operation that performs an exponentially precise energy measurement and $\alpha'=1/\poly(n)$. More precisely, $U_{EM}$ acts on an eigenstate $\ket{\psi_E}$ and additional ancillas as
\begin{equation}
U_{EM}\ket{\psi_E,0,0}=\ket{\psi_E}\sum_{E'} a_{E'} \ket{E',g(E')}
\end{equation}
where $E', g(E')$ live in a poly-size register, $E'$ is the measurement outcome and $g(E')$ is some garbage data; furthermore, the probability of observing $E'$ obeys Eq.~\eqref{eq:measacc} where $\delta =1/2^{\Omega(n)}$ and $\eta=1-1/\poly(n)$. 

It can be seen that since $\tilde{U}_{EM}$ is close to $U_{EM}$ in operator norm ($||\tilde{U}_{EM}-U_{EM}||\leq 1/\poly(n)$), the total variation distance between the probability distributions resulting from a measurement of the output of $U_{EM}$ and $\tilde{U}_{EM}$ is also bounded by $\beta=1/\poly(n)$.  Hence, the ability to exponentially fast-forward a Hamiltonian implies the ability to generate a quantum circuit that solves the $\beta$-approximate energy sampling problem with confidence $\eta$ sampling error $\beta=1/\poly(n)$ and resolution $\delta =1/2^{\Omega(n)}$.  
\section{Proof of Theorem 1}\label{app:proofsthm}
In this section we give technical proof of Theorem~\ref{thm:superres} of the main text.
\setcounter{thm}{0}
\begin{thm}[\bf Quantum algorithm for super-resolution energy measurements]\label{thm:superres} 
Consider any quantum diagonalizable Hamiltonian $\tilde{H}=U^{\dagger}H_f U$ as in (\ref{eq:specialHam}). Then, the following quantum algorithm efficiently solves the $\beta-$approximate Energy Sampling problem for Hamiltonian $\tilde{H}$, with the initial state $\ket{\psi}$ and parameters $\eta=1$ and $\delta =2^{-l}$:
\begin{itemize}
\item Query a $\beta$-approximate sampler for $U$, with initial state $\ket{\psi}$.
\item Given an outcome $z$, output an $l$-digit approximation of the value $f(z)$ .
\end{itemize} 
\end{thm}
\emph{Proof:}
First let us consider this algorithm in the case we have access to an \emph{exact sampler} from $U$, that is, we take $\beta=0$. Such sampler outputs $z$ with probability $P_z=|\bra{z}U\ket{\psi}|^2$. We denote the function that approximates $f(z)$ to $l$-bits as $\tilde{f}(z)$, implying that $|\tilde{f}(z)-f(z)|\leq \delta =2^{-l}$. Assuming the values of $f(z)$ lie in the interval $[0,1]$, the function $\tilde{f}(z)$ outputs values $E_m\in \{0,\delta ,....,1-\delta,1 \}$. 

Let us denote the probability of outputting $E_m$, via the procedure described in the theorem, as $q_m$. Then 
\begin{equation}\label{eq:qm}
q_m= \sum_{z\in \tilde{f}^{-1}(E_m)}P_z, 
\end{equation}
where $\tilde{f}^{-1}(E_m)$ is the pre-image of $E_m$ under the function $\tilde{f}$ i.e., the set of values $z$ that are mapped to $E_m$ via $\tilde{f}$.

Let us demonstrate that this probability distribution obeys the constraints given by Eq.~\eqref{eq:Ensamp}, of an energy sampler with $\epsilon=0$ and $\delta =2^{-l}$. Let us define $f^{-1}([E_a,E_b])$ as the pre-image of the energy interval $[E_a,E_b]$ under $f(z)$ i.e.,
\begin{equation}
f^{-1}([E_a,E_b])=\{z ~|~ f(z)\in[E_a,E_b]\}.
\end{equation}
Given an outcome value $z\in f^{-1}([E_a,E_b])$, we have that $\tilde{f}(z)\in [E_a-\delta ,E_b+\delta ]$. Hence, the probability that Procedure~1 outputs a value $E'\in [E_a-\delta ,E_b+\delta ]$ 
\begin{align}
\Pr(E'\in [E_a-\delta ,E_b+\delta ])&\geq\sum_{z\in f^{-1}([E_a,E_b])} P_z\\
& = \sum_{z\in f^{-1}([E_a,E_b])} |\bra{z}U\ket{\psi}|^2 \label{eq:probab}
\end{align}
On the other hand, we have that the eigenstates of $H$ are given by $U^{\dagger} \ket{z}$ with eigenvalue $f(z)$. It follows that the spectral projection of $H$ in an interval $[E_a,E_b]$ is given by
\begin{equation}\label{eq:spec_proj}
\Pi_{[E_a,E_b]}= \sum_{z\in f^{-1}([E_a,E_b])}U^{\dagger} \ket{z}\bra{z}U.
\end{equation}  
Consequently, defining $\rho=\ket{\psi}\bra{\psi}$ and using Eqs.~\eqref{eq:probab} and \eqref{eq:spec_proj} we have that the probability that Procedure~1 outputs the energy value $E'$ is given by 
\begin{equation}
\Pr(E'\in [E_a-\delta ,E_b+\delta ])\geq \text{tr}\left(\rho \Pi_{[E_a,E_b]} \right),
\end{equation}
which is an energy sampler with parameters $\delta =2^{-l}$ and $\epsilon=0$. 

Let us now consider the more general case where Procedure~1 has access to a $\beta$-approximate sampler for $U$ i.e., the outcome $z$ is observed with probability $P'_z$ such that 
\begin{equation}
\sum_z |P'_z-P_z|\leq \beta.
\end{equation}
In this case, analogously to Eq.~\eqref{eq:qm}, we define $q'_m$ as the probability that the procedure described in the theorem outputs $E_m=m \delta $, which is given by 
\begin{equation}\label{eq:qprimem}
q'_m= \sum_{z\in \tilde{f}^{-1}(E_m)}P'_z.
\end{equation}
Using Eq.~\eqref{eq:qm}, we have that 
\begin{align}
\sum_m|q_m-q'_m|&= \sum_m \left|\sum_{z\in \tilde{f}^{-1}(E_m)} (P_z-P'_z)\right|\\
				&\leq \sum_m \sum_{z\in \tilde{f}^{-1}(E_m)} |P_z-P'_z|\\
				&=\sum_z |P_z-P'_z| \leq \beta.
\end{align}
This shows that the probability distribution $\{q'_m\}$ has a total variation distance of at most $\beta$ with respect to the energy sampler with parameters $\epsilon=0$ and $\delta =2^{-l}$ defined by the probabilities $\{q_m\}$. \QEDB
\section{Locality of Hamiltonians diagonalized by IQP circuits with bounded degree} \label{app:hamcalculations}
Let us consider the class of Hamiltonians $H_{IQP}=\sum_j w_j U_{IQP}^{\dagger}{\hat{n}}_l U_{IQP}$, where $U_{IQP}$ is an IQP circuit~\cite{IQPapprox}, of the form
\begin{equation}\label{U_IQP}
U_{IQP}=\exp(i \frac{\pi}{8}\sum_{(j,k)\in E_{G}} w_{jk} X_j X_k + \sum_k v_k X_k).
\end{equation}
where $E_{G}$ denotes the edges of the interaction graph of the circuit.   
Clearly, $U_{2D}$ is a particular case of this more general set of unitaries, which is obtained when the weights are $w_{jk}=2$ if $(j,k)$ corresponds to an edge of a 2D lattice.
We start by calculating $U_{IQP}^{\dagger}Z_l U_{IQP}$. We can write $U_{IQP}=\exp(i\pi/8 H_{XX})$ with 
\begin{align}\label{eq:HXX}
H_{XX}=& \sum_{(j,k)\in E_{G}} w_{jk} X_j X_k + \sum_k v_k X_k \\
=& X_l\otimes(v_l\mathbb{I}+\sum_{k\neq l} w_{kl} X_k)+\bar{H}_l\\
=&X_l\otimes H_l+\mathbb{I}_l\otimes\bar{H}_l
\end{align}
where we have defined the Hamiltonians $H_l$ and $\bar{H}_l$ acting on the $n-1$ qubits other than $l$ as
\begin{align}
H_l &= v_l\mathbb{I}+\sum_{k\neq l} w_{kl} X_k\\
\bar{H}_l & =\sum_{j,k \neq l } w_{jk} X_j X_k + \sum_{k\neq l} v_k X_k
\end{align}
Using this, we can write 
\begin{align}
U_{IQP}^{\dagger}Z_l U_{IQP}&= e^{-i\frac{\pi}{8} X_l\otimes H_l} Z_l~e^{i\frac{\pi}{8} X_l\otimes H_l}\\
&=Z_l -i\frac{\pi}{8}[X_l\otimes H_l , Z_l]\\
&~-\frac{1}{2}\left(\frac{\pi}{8}\right)^2 [X_l\otimes H_l ,[X_l\otimes H_l , Z_l]+... \nonumber \\
&=\sum_{k=0}^{\infty}\frac{C_k}{k!}\left(-i\frac{\pi}{8}\right)^k,
\end{align} 
where $C_0=Z_l$ and $C_k$ results from applying $k$-times the commutator $[X_l\otimes H_l ,~\cdot~]$ to the operator $Z_l$. By calculating the first few commutators, a pattern can be noticed
\begin{align}
C_1=[X_l\otimes H_l, Z_l]=[X_l, Z_l]\otimes H_l=-2i Y_l\otimes H_l,
\end{align}
\begin{align}
C_2&=-2i [X_l\otimes H_l, Y_l\otimes H_l]\\ &=-2i [X_l, Y_l]\otimes H_l^2\\
&= -2i(2i) Z_l \otimes H_l^2. 
\end{align}
The even terms are thus given by 
\begin{align}
C_{2k}&= (2i)^k(-2i)^k Z_l\otimes H_l^{2k}\left(-i\frac{\pi}{8}\right)^{2k}\\
&= (-1)^k Z_l\otimes \left(\frac{\pi}{4}H_l\right)^{2k}, 
\end{align}
whereas the odd terms yield
\begin{align}
C_{2k+1}&= (2i)^k(-2i)^{k+1} Y_l\otimes H_l^{2k+1}\left(-i\frac{\pi}{8}\right)^{2k+1}\\
&=- (-1)^k Y_l\otimes \left(\frac{\pi}{4}H_l\right)^{2k+1}, 
\end{align}
Using these results, we can write 
\begin{align}
U_{IQP}^{\dagger}Z_l U_{IQP}= Z_l\otimes \cos\left(\frac{\pi}{4}H_l\right)-Y_l\otimes \sin\left(\frac{\pi}{4}H_l\right).
\end{align}
This term acts non-trivially in $d+1$ qubits, where $d$ is the number of non-zero values of $w_{kl}$, for $k\neq l$ i.e., the number of qubits that interact with qubit $l$ via Hamiltonian $H_{XX}$ in Eq.~\eqref{eq:HXX}. In the case discussed in the main text, $H_{XX}$ is defined on a 2D lattice, which implies that each qubit interacts with 4 other qubits. Hence, the Hamiltonian $H_{2D}$ from Eq.~\eqref{eq:H2d} is 5-local. 
In fact, since all the weights are the same ($w_{jk}=2, (j,k)\in E_{2D}$) the expression above simplifies to
\begin{equation}
U_{2D}^{\dagger}Z_l U_{2D}=Z_l \prod_{j:(j,l)\in E_{2D}} X_j.
\end{equation}

\section{Proof and consequences of Lemma~\ref{lemma:ObstructionBMS}}\label{app:Stockmeyer}

\setcounter{lem}{1}
\begin{lem}[\textbf{Stockmeyer error}]
Let $\mathcal{Q}_n,{n\in\mathbb{N}}$ be a family of uniformly-generated poly-size $n$-qubit quantum circuits  with $m$ output bits and the hiding property
$$\forall U \in \mathcal{Q}_n, x\in\{0,1\}^m, \exists  U_x\in \ \mathcal{Q}_n : q_U(x)=q_{U_x}(0_m).$$
Assume there exists a classical algorithm $\mathcal{A}$ that samples from $q_U$ with $\ell_1$-error $\beta$ in $O(\poly(n))$ time for any circuit $U\in\mathcal{Q}_n$. Then, for any $0<\nu<1$, there is an $\textnormal{FBPP}^\textnormal{NP}$ algorithm which, given access to $\mathcal{A}$, approximates $q_U(x),x\in\{0,1\}^m$ up to additive error $\varepsilon$
\begin{equation}
\varepsilon\in O \left( \frac{q_U(x)}{\poly(n)}+ \frac{\beta}{2^m\nu}\left(1+\frac{1}{\poly(n)} \right)\right).\notag
\end{equation}
with probability $1-\nu$ over the choice of~$x\in\{0,1\}^m$. 
\end{lem}

\begin{proof}
For any $U\in\mathcal{Q}_n$, let $p_U$ the distribution generated by $\mathcal{A}$ fulfilling
\begin{equation}
\|p_U-q_U\|=\sum_{x\in\{0,1\}^m} |p_U(x)-q_U(x)|< \beta.
\end{equation}
As discussed in \cite{IQPapprox}, Stockmeyer's algorithm implies the existence of an FBPP$^\textrm{NP}$ algorithm that computes a relative-error estimate $\tilde{p}_U(x)$ of ${p}_U(x)$:
\begin{equation}
| \tilde{p}_U(x) - p_U(x)| \leq  \frac{p_U(x)}{\poly(n)}.
\end{equation}
Using the triangle inequality we get
\begin{align}\label{eq:error_estimate}
| \tilde{p}_U(x) - q_U(x)| \leq & \frac{q_U(x)}{\poly(n)}\\&+ | p_U(x) - q_U(x)|\left(1+\frac{1}{\poly(n)}\right).\notag
\end{align}
Last, for any $0< \nu < 1$, Markov's inequality implies that 
\begin{align}\label{eq:StockError}
| \tilde{p}_U(x) - q_U(x)|\leq & \frac{q_U(x)}{\poly(n)}+ \frac{\beta}{2^m\nu}\left(1+\frac{1}{\poly(n)}\right)
\end{align}
with probability $1-\nu$ over the choice of~$x\in\{0,1\}^m$. This completes the proof.
\end{proof}

We review a few consequences of this lemma that are mentioned in the main text, section \ref{sec:conjectures}.

(I) \emph{Classical hardness of near-exact quantum sampling problems based on the non-collapse of the Polynomial Hierarchy.} First, we review how Stockmeyer's algorithm can be used to rule out near-exact classical simulations of certain sampling problems (where $\beta=1/2^\poly(n)$), assuming only the non-collapse of the Polynomial Hierarchy (PH) \cite{bosonsampling,IQPexact}. This is the case for poly-size quantum circuits with output probabilities that are $\#$P-hard to compute up to a relative error, even when the number of output bits $m$ is constant. 

To demonstrate this, it is important to note that the output probabilities of the latter are either zero or larger than  $p^*=1/2^{O{(n^c)}}$, for some constant $c$~\cite{Kuperberg15}. Using this fact, we can estimate these probabilities up to relative error via Stockmeyer's algorithm by choosing a value of $\beta<p^*/\poly(n)$ i.e., significantly smaller than the probability gap $p^*$. This can be seen by analysing the error of the estimation on the right hand side (RHS) of Eq.~\eqref{eq:error_estimate}. If $q_U(x)$ is not zero, then it is larger than $p^*$~\cite{Kuperberg15} and the aforementioned choice of $\beta$ guarantees that the error on the RHS of Eq.~\eqref{eq:error_estimate} is upper bounded by $q_U(x)/\poly(n)$. Hence, in this case, $\tilde{p}_U(x)$ is a relative error estimation of $q_U(x)$. On the other hand, if $q_U(x)=0$, the error on the RHS of Eq.~\eqref{eq:error_estimate} is upper bounded by $p^*/\poly(n)$. This implies that $\tilde{p}_U(x)\leq p^*/\poly(n)$, which is significantly smaller than $p^*$. This allows us to conclude that $q_U(x)=0$. 

Consequently, the existence of an efficient classical algorithm for near-exact sampling of these quantum circuits implies, via  the previous arguments, the existence of an algorithm in the complexity class FBPP$^\textnormal{NP}$ that computes $\#$P-hard to estimate output probabilities.  It follows that PH collapses to its 3rd level, since it is known that FBPP$^\textnormal{NP}$ is in level three, and $\textnormal{P}^{\#\textnormal{P}}$ is above the hierarchy (by Toda's theorem) \cite{Toda:1991:PHP:125944.125952}.

(II) \emph{Classical hardness of approximate quantum sampling problems based on additional complexity-theoretic conjectures.} As discussed in section \ref{sec:conjectures}, it is possible to extend the above results to rule out classical simulations with constant or inverse polynomial sampling errors assuming additional conjectures. Specifically, Refs.\mot\cite{bosonsampling,IQPapprox} exploit Lemma \ref{lemma:ObstructionBMS}, in the case $m=n$, to prove the hardness of approximate sampling problems based on three conjectures: the non-collapse of PH, anticoncentration, and the average-case $\#$P-hardness of approximating output probalities of a quantum device. We review the key idea behind this proof. If the distribution $q_U(x)$ anticoncentrates (assumption C3 in section \ref{sec:conjectures}), then  $\mathrm{prob}_{x} \left[q_U(x) > \alpha/2^n\right] > \gamma$, for some constants $\alpha, \gamma\in O(1)$. Then, with constant probability, the error in equation (\ref{eq:StockError})  is a relative error for $q_U(x)$, if $m=n$. It follows that there is an FBPP$^\textnormal{NP}$ algorithm that can approximate the output probabilities of the device up to a relative error for a constant fraction of the instances (i.e., in ``average'' when we randomize over the choice of probability). If we assume this problem to be $\#$P-hard  (the average-case assumption C2 in section \ref{sec:conjectures}), then the Polynomial Hierarchy collapses to its 3rd level.

(III) \emph{The $m < \log(n)$ case.} 
 In this scenario, the error in the right hand side of Eq.~(\ref{eq:StockError}) is $\Omega(1/\poly{(n)})$. Assuming the existence of an efficient classical sampler of the circuit family $\mathcal{Q}_n$, the Stockmeyer argument implies the existence of an average-case FBPP$^\textnormal{NP}$ algorithm that approximates up to this error the output probabilities of\mot$\mathcal{Q}_n$. This error is quite large and, in fact, can be achieved simply by querying the hypothetical classical sampler a polynomial number of times i.e., this problem would be in BPP.
 
 In order to draw an unlikely complexity theoretic implication in this scenario, one would have to prove that approximating a typical output probability of a quantum circuit to $1/\poly{(n)}$ errors is hard for a complexity class which is unlikely to be contained in BPP. However, it is important to note that this approximation problem can be efficiently solved by sampling from quantum circuits, and is therefore in BQP. It is thus implausible that one can show that this problem is  \#P-hard, or even  NP-hard, for then quantum computers would be able to solve such problems; which is, in turn, considered to be unlikely \mot\cite{Bennet97StrenghtsWeaknnesses_QC,Aaronson-ProcRS-2005}. Hence, new techniques seem to be required to give complexity theoretic evidence for the classical harness of approximate sampling problems with small output space. 
\end{document}